\documentclass[letterpaper, journal]{IEEEtran}
\usepackage{amsmath,amsfonts}
\usepackage{array}
\usepackage[caption=false,font=normalsize,labelfont=sf,textfont=sf]{subfig}
\usepackage{textcomp}
\usepackage{stfloats}
\usepackage{url}
\usepackage{verbatim}
\usepackage{graphicx}
\usepackage{cite}

\usepackage{hyperref}
\usepackage{amsmath}
\usepackage{amsthm}
\usepackage{amssymb}
\usepackage{mathtools}
\usepackage{graphicx}
\usepackage{enumitem}
\usepackage{caption}
\usepackage{algorithm}
\usepackage{algpseudocode}
\usepackage{comment}
\usepackage{ragged2e}
\usepackage{xcolor}
\justifying
\newtheorem{proposition}{Proposition}
\newtheorem{corollary}{Corollary}
\newtheorem{lemma}{Lemma}

\algrenewcommand\algorithmiccomment[1]{\hfill\(\triangleright\) #1}
\algrenewcommand{\alglinenumber}[1]{\footnotesize #1} 
\algrenewcommand{\algorithmicindent}{1.0em} 

\begin{document}

\title{New Constructions of Polar Code Based on Refined Error Probability Analysis}

\author{Hassan Noghrei and Murad Abdullah}



\maketitle

\begin{abstract}


This paper presents a refined analysis of the block error rate (BLER) of polar codes over symmetric binary-input discrete memoryless channels under successive cancellation (SC) and SC list (SCL) decoding. A novel expression for the BLER of a polar code under SC decoding is derived directly in terms of the decoder’s LLRs. Building on this formulation, we propose a polar code construction algorithm optimized for SC decoding and evaluate its performance under SC and dynamic SC flip (DSCF) decoding against established SC-optimized constructions, including Gaussian approximation (GA)-based and Tal-Vardy polar codes. 
Furthermore, by decomposing the BLER into path loss and path selection components, we derive a novel LLR-based expression for the path loss probability, which enables an SCL-optimized polar code construction method. The proposed constructions are evaluated under SCL decoding with list sizes \(2\), \(4\), and \(8\), and are compared with 5G standard polar codes, GA-based designs, and Reed-Muller polar codes.
Simulation results show that the proposed SC-optimized polar codes achieve up to a \(0.2\)~dB performance gain under DSCF decoding over the AWGN channel compared to benchmark constructions, and exhibit superior performance over binary symmetric channels. For SCL-optimized polar codes, the proposed method achieves comparable or improved performance across all considered list sizes, with gains of up to \(0.4\)~dB relative to benchmark designs.

\end{abstract}


\section{Introduction}


\IEEEPARstart{T}{he} evolution toward sixth-generation (6G) wireless systems is expected to support mission-critical applications such as telesurgery and industrial automation~\cite{wang2023road,tataria20216g}. In 5G systems, ultra-reliable and low-latency communications (URLLC) target success probabilities of \(0.99999\) for short packets with end-to-end latencies around \(1\,\mathrm{ms}\), which suffice for many industrial and safety-critical use cases~\cite{chen2018ultra,Nokia2016MissionCritical}. Looking ahead, 6G extreme URLLC (xURLLC) imposes significantly stricter requirements, including sub-millisecond latency and BLER levels as low as \(10^{-7}\)–\(10^{-9}\), to enable highly demanding real-time control and safety-critical applications~\cite{wang2023road,chen2018ultra}.


Meeting these operating points places stringent demands on the physical layer. In particular, channel coding, which directly impacts both latency and reliability, must provide ultra-high reliability with minimal decoding delay~\cite{zhang2023channel,rowshan2024channel,miao2024trends}. Among candidate schemes, polar codes~\cite{arikan2009channel} are especially attractive due to their flexible block lengths, strong finite-length performance under CRC-aided SCL decoding, and proven practicality through their adoption in 5G NR control channels~\cite{bioglio2020design}.



Polar codes are the first class of error-correcting codes proven to achieve the symmetric capacity of B-DMCs~\cite{arikan2009channel}. They rely on channel polarization, which converts independent channel uses into synthesized bit-channels that become either highly reliable or highly unreliable as the block length grows. By transmitting information over reliable bit-channels and freezing the rest, polar codes achieve capacity asymptotically under SC decoding.

SC decoding is appealing due to its low complexity and hardware efficiency~\cite{tong2023fast}, but its performance is limited at short and moderate block lengths~\cite{tal2015list,hashemi2017fast,balatsoukas2015llr}. To address this drawback, several enhanced SC-based decoders have been proposed, including SCF~\cite{afisiadis2014low}, DSCF~\cite{chandesris2018dynamic}, and SCL decoding~\cite{tal2015list}, which substantially improve error-correction performance.

The performance of a polar code critically depends on the selection of reliable bit-channels for information transmission, a process known as polar code construction. This process is inherently channel- and decoder-dependent. For the binary erasure channel (BEC) under SC decoding, Arıkan~\cite{arikan2009channel} derived exact Bhattacharyya parameters and sub-channel capacities to determine information and frozen bit positions. For the AWGN channel, he further proposed an MC-based approach to estimate bit-channel error probabilities and select the most reliable ones for information transmission.

To improve construction accuracy, Mori and Tanaka~\cite{mori2009performance} introduced a density evolution (DE)-based method that tracks LLR distributions to estimate bit error probabilities, albeit at high computational cost. Tal and Vardy~\cite{tal2013construct} later proposed a degradation–upgrading technique that efficiently approximates these probabilities for AWGN channels under SC decoding. To further reduce complexity, Trifonov~\cite{trifonov2012efficient} developed a Gaussian Approximation (GA)-based construction that models LLRs as Gaussian random variables, maintaining high accuracy while significantly reducing the computational complexity of polar code construction optimized for SC decoding over AWGN channels.

Later, the Polarization Weight (PW) method~\cite{he2017beta} was introduced as a channel-agnostic polar code construction technique, in which bit-channel reliability is determined solely by their indices rather than by specific channel characteristics. In the 5G standard, the reliability order of polar code bit-channels is defined using a fixed sequence obtained through extensive computer-based searches and simulations~\cite{bioglio2020design}.

Recent machine learning-based approaches, including reinforcement learning (RL) and genetic algorithms (GenAlg), have been explored for polar code construction. RL methods formulate bit-channel selection as a Markov decision process, while GenAlg approaches evolve decoder-aware constructions through genetic operations, enabling optimization for specific decoders such as SCL~\cite{liao2021construction,huang2019ai,elkelesh2019decoder,zhou2021low}.


Most existing polar code constructions rely on DE or its approximations to estimate bit-channel reliability. Although computationally efficient, DE-based methods implicitly assume error-free previously decoded bits, corresponding to a genie-aided decoder and neglecting error propagation in practical decoding~\cite{mori2009performance}. Moreover, these constructions are inherently channel-dependent, as reliabilities are computed for a specific channel model; consequently, codes optimized for AWGN may suffer performance degradation over other channels such as the BSC~\cite{mori2009performance,tal2013construct}. 
For more general channels, MC- and GenAlg-based approaches can yield more accurate reliability estimates, but at the cost of a large number of samples and high computational complexity~\cite{liao2021construction,vangala2015comparative}. In addition, most existing constructions are tailored to SC decoding and do not fully exploit the potential of enhanced SC-based decoders such as SCF, DSCF, or SCL.

In this work, we investigate the optimization of polar codes over arbitrary symmetric B-DMCs under SC and SCL decoding. We derive novel LLR-based approximations of the BLER for both decoders that closely match simulation results. Based on these formulations, we develop a recursive construction algorithm that selects information and frozen bits for arbitrary code lengths and rates, tailored to the target decoding algorithm. We further show that the SC BLER formulation can be exploited to reduce the complexity of DSCF decoding. Simulation results validate the effectiveness of the proposed constructions under SC, DSCF, and SCL decoding. Our main contributions are summarized as follows:

\begin{enumerate}
\item \textit{Novel BLER expression for SC decoding.}  
We derive a novel LLR-based expression for estimating the BLER of polar codes under SC decoding. The proposed formulation is fully channel-agnostic and applicable to arbitrary symmetric B-DMCs. It explicitly captures the sequential decision of SC decoding, accounting for the propagation of early decoding errors to subsequent stages. The resulting BLER estimates closely match simulation results.

\item \textit{SC-optimized polar code construction.}  
Building on the proposed BLER formulation, we develop a recursive polar code construction algorithm that selects information and frozen bits according to their estimated reliabilities. The method supports arbitrary code lengths over symmetric B-DMCs. Over the AWGN channel, the resulting codes outperform GA-based and Tal--Vardy designs under DSCF decoding, while matching or exceeding SC-optimized constructions under SC decoding. Over the BSC channel, they significantly outperform Bhattacharyya-based and GA-based designs under both SC and DSCF decoding.


\item \textit{A novel BLER expression for SCL decoding.} 
We further analyze the BLER of polar codes under SCL decoding with a fixed list size by decomposing it into path loss and path selection error probabilities. The path loss term accounts for the elimination of the correct path during list pruning, while the path selection term corresponds to selecting an incorrect candidate despite the correct path surviving. Based on this decomposition, we derive an LLR-based expression for the path loss probability, which quantifies the likelihood of eliminating the correct path at each information bit. This formulation serves as the foundation for an SCL-optimized polar code construction tailored to a given list size.

    \item \textit{SCL-optimized polar code construction.}  
    Based on the proposed path loss error formulation for SCL decoding, we develop a recursive polar code construction algorithm optimized for SCL decoding with a given list size, which freezes bit-channels with high path loss error probabilities. Simulation results demonstrate that the resulting SCL-optimized polar codes achieve comparable or superior performance compared with benchmark constructions, including the 5G polar code, GA-based designs, and Reed-Muller polar codes, across all considered list sizes.
\end{enumerate}

This paper extends our previous work presented in~\cite{noghrei2025new}, in which we introduced a new formulation of the BLER under SC decoding and proposed a polar code construction method optimized for SC decoding. In this paper, we further develop this framework by providing additional insights, a more detailed analysis, and extended simulation results. Moreover, we derive a new formulation for the BLER under SCL decoding and propose a corresponding SCL-optimized polar code construction algorithm.

The remainder of this paper is organized as follows. Section~\ref{sec:preliminary} reviews the fundamentals of polar codes and the decoding algorithms considered in this work. Section~\ref{sec:bep} presents the proposed analytical formulation for estimating the BLER under SC decoding. Section~\ref{sec:SC_cons} describes the recursive polar code construction algorithm optimized for SC decoding. Section~\ref{sec:scl_bep} introduces the proposed analytical formulation for estimating the path loss error probability under SCL decoding. Section~\ref{sec:scl_cons} presents the recursive polar code construction algorithm optimized for SCL decoding. Section~\ref{sec:sim} reports simulation results for different channel models and analyzes the computational complexity of the proposed algorithms. Finally, Section~\ref{sec:con} concludes the paper and outlines potential directions for future research.

\section{PRELIMINARIES}\label{sec:preliminary}

This section briefly reviews the fundamentals of polar codes, which underpin the subsequent analysis.

\subsection{Notations}

Throughout this paper, calligraphic letters (e.g., $\mathcal{X}$ and $\mathcal{Y}$) denote sets. The Cartesian product of $\mathcal{X}$ and $\mathcal{Y}$ is written as $\mathcal{X} \times \mathcal{Y}$, and $\mathcal{X}^n$ denotes the $n$-fold Cartesian power. The set of consecutive integers from $a$ to $b$ is denoted by $[\![a,b]\!]$. Vectors are denoted by $v_1^N = (v_1,\dots,v_N)$, and $v_i^j$ represents the subvector $(v_i,\dots,v_j)$. For an index set $\mathcal{I} \subseteq [\![1,N]\!]$ with complement $\mathcal{I}^c$, the subvectors $v_{\mathcal{I}}$ consists of elements indexed by $\mathcal{I}$. Matrices are denoted by bold uppercase letters; for example, $\mathbf{G}_{N\times M}$ and $\mathbf{G}_N$ represent $N\times M$ and $N\times N$ matrices, respectively. 

Consider a B-DMC $\mathcal{W}:\mathcal{X}\to\mathcal{Y}$ with $\mathcal{X}=\{0,1\}$. The transition probability is $\mathcal{W}(y\mid x)$. If $\mathcal{W}$ is symmetric, there exists a permutation $\pi_1$ on $\mathcal{Y}$ such that \(\mathcal{W}(y\mid1)=\mathcal{W}(\pi_1(y)\mid0),\quad \forall\,y\in\mathcal{Y}.\)
Let $\pi_0$ be the identity permutation and define $\pi_x(y)\triangleq x\cdot y$. Then, for all $a,x\in\mathcal{X}$ and $y\in\mathcal{Y}$, the channel symmetry condition implies that \(\mathcal{W}(y\mid x\oplus a)=\mathcal{W}(a\cdot y\mid x).\)
This action extends elementwise to vectors: for $x_1^N\in\mathcal{X}^N$ and $y_1^N\in\mathcal{Y}^N$,
\[
x_1^N\cdot y_1^N \triangleq (x_1\cdot y_1,\dots,x_N\cdot y_N).
\]
We use $(\mathcal{X}^N \times \mathcal{Y}^N, \mathbb{P})$ to define the probability space with probability measure $\mathbb{P}$ defined as in Section~V of~\cite{arikan2009channel}. On this space, we define the random vectors $(U_1^N, X_1^N, Y_1^N, \hat{U}_1^N)$, which represent the encoder input, the channel input to $\mathcal{W}^N$ (consisting of $N$ i.i.d.\ uses of $\mathcal{W}$), the channel output, and the decoder decisions, respectively. For each realization $(u_1^N, y_1^N) \in \mathcal{X}^N \times \mathcal{Y}^N$, these random vectors take the values $U_1^N = u_1^N$, $X_1^N = u_1^N \mathbf{G}_N$, and $Y_1^N = y_1^N$, while the decoder output $\hat{U}_1^N$ is generated sequentially according to the SC decoding rule.

\subsection{Polar Codes}

A $(N,K)$ polar code of length $N=2^n$ and rate $R=K/N$ is a linear block code constructed via channel polarization over $N$ independent uses of a B-DMC $\mathcal{W}$. Channel polarization produces a set of sub-channels (\(\mathcal{W}_N^{(i)}:\mathcal{X} \to \mathcal{Y}^N \times \mathcal{X}^{i-1}\), where \(\mathcal{X}^0 \triangleq \emptyset\)) that are partitioned into an information set $\mathcal{I}$ and its complement $\mathcal{I}^c$, which contains the so-called frozen bits. The $K=|\mathcal{I}|$ most reliable sub-channels carry information bits, while the remaining sub-channels correspond to frozen bits, which are fixed (typically to zero) and known to both the transmitter and receiver.

 A binary source vector \( u_1^N \), consisting of \(K\) information bits and \(N - K\) frozen bits, is encoded into a codeword \( x_1^N \) as
\(x_1^N = u_1^N \mathbf{G}_N\),
where \( \mathbf{G}_N \) denotes the generator matrix of the polar code. The matrix \( \mathbf{G}_N \) is recursively defined as
\(\mathbf{G}_N = \mathbf{F}_2^{\otimes n}\), where \( \otimes n \) denotes the \(n\)-th Kronecker power, and \(\mathbf{F}_2 = 
\begin{bmatrix}
1 & 0 \\
1 & 1
\end{bmatrix}\)
is the basic \(2 \times 2\) kernel matrix.

At the receiver, the decoder estimates the vector $\hat{u}_1^N$ from the received signal $y_1^N$ using log-LLRs. SC decoding proceeds sequentially, estimating each bit $\hat{u}_i$ based on $y_1^N$ and the previously decoded bits $\hat{u}_1^{i-1}$. The decision LLR for bit $u_i$ is given by
\begin{equation}\label{eq:llr_out}
L_i = \ln \left( \frac{\Pr\left(U_i = 0 \mid Y_1^N =y_1^N,U_1^{i-1} =\hat{u}_1^{i-1}\right)}{\Pr\left(U_i = 1 \mid Y_1^N =y_1^N, U_1^{i-1} =\hat{u}_1^{i-1}\right)} \right),
\end{equation}
and the hard decision rule is
\begin{equation}
\hat{u}_i =
\begin{cases}
0, & \text{if } i \in \mathcal{I}^c \text{ or } L_i \ge 0, \\
1, & \text{otherwise}.
\end{cases}
\end{equation}
Although SC decoding is asymptotically optimal, its performance degrades at short-to-moderate block lengths. To address this limitation, several enhanced SC-based decoders have been proposed, such as SCL~\cite{tal2015list}, SCF~\cite{afisiadis2014low}, and DSCF~\cite{chandesris2018dynamic}.


SCL decoding alleviates the error propagation of SC decoding by maintaining a list of candidate decoding paths. At each information bit position, every path is extended over both bit values $0$ and $1$, while frozen bits are fixed to zero. When the number of candidate paths exceeds the list size $L$, the least reliable paths are pruned based on their path metrics.
At the decoding stage $i$, the path metric of the $l$-th path is defined as \cite{balatsoukas2015llr}
\begin{equation}\label{eq:path_metr}
\mathrm{PM}[l]_i = \sum_{j=1}^{i} \ln \!\left( 1 + e^{-(1 - 2\hat{u}[l]_j)L[l]_j} \right),
\end{equation}
where $\hat{u}[l]_j$ and $L[l]_j$ denote the estimated bit and the corresponding LLR, respectively. The decoder retains the $L$ paths with the smallest path metrics~\cite{tal2015list}.
Although SCL decoding significantly improves error-correction performance, its computational complexity scales as $\mathcal{O}(L N \log N)$~\cite{sarkis2015fast}, which is higher than that of SCF and DSCF decoding schemes~\cite{afisiadis2014low,chandesris2018dynamic}.

    The SCF decoder enhances the error-correction capability of SC decoding by selectively correcting unreliable bit decisions. When SC decoding fails, as detected by a CRC test, SCF identifies the most unreliable bit (typically the one with the smallest absolute LLR) and flips its decision in a subsequent SC decoding attempt~\cite{afisiadis2014low}. 
The DSCF decoder extends SCF by dynamically selecting bit-flipping positions using an optimized reliability metric. Instead of flipping the least reliable bit, DSCF identifies the bit position that maximizes the likelihood of correcting the SC decoding error. Additional decoding attempts are initiated only when the initial SC decoding fails the CRC test. For bit-flip order $\omega=1$, the flipped position is selected as~\cite{chandesris2018dynamic}
\begin{equation}\label{eq:dscf_met}
i^* = \arg \max_{i\in \mathcal{I}} \left( e^{-\alpha |L_i|}
\prod_{\substack{j \in \mathcal{I} \\ j \le i}}
\frac{1}{1 + e^{-\alpha |L_j|}} \right),
\end{equation}
where $\alpha$ is a scaling parameter, set to $\alpha=0.3367$ in this work following~\cite{chandesris2018dynamic}.

\subsection{Error Performance of Polar Codes}

The BLER $P(\mathcal{E})$ of a $(N,K,\mathcal{I})$ polar code is defined as the probability of the block error event
\[
\mathcal{E} \triangleq \bigl\{ (u_1^N,y_1^N)\in\mathcal{X}^N\times\mathcal{Y}^N : u_{\mathcal{I}}\neq \hat u_{\mathcal{I}} \bigr\}.
\]
Following~\cite{arikan2009channel}, $\mathcal{E}$ can be decomposed into disjoint events $\mathcal{E}_i$, where $\mathcal{E}_i$ denotes the event that the first decoding error occurs at bit $i\in\mathcal{I}$. Each $\mathcal{E}_i$ is contained in
\[
\mathcal{B}_i \triangleq \bigl\{ (u_1^N,y_1^N)\in\mathcal{X}^N\times\mathcal{Y}^N : u_i \neq \hat u_i  \bigr\},
\]
which corresponds to an error at bit $i$ regardless of prior decisions. Consequently, the BLER admits the upper bound
\begin{align}\label{eq:ar_bound}
P(\mathcal{E})=\sum_{i\in\mathcal{I}} P(\mathcal{E}_i)
\le \sum_{i\in\mathcal{I}} P(\mathcal{B}_i)
\le \sum_{i\in\mathcal{I}} Z\!\left(\mathcal{W}_N^{(i)}\right),
\end{align}
where $Z(\mathcal{W}_N^{(i)})$ is the Bhattacharyya parameter of the $i$-th bit-channel.
This bound is widely used for polar code construction to identify reliable subchannels~\cite{arikan2009channel,tal2013construct}. Building on this formulation, the next section derives a refined BLER expression for improved polar code construction.

\section{Block error rate Under SC decoding}\label{sec:bep}

In this section, we derive new formulations for the single-bit error probabilities \(P(\mathcal{E}_i)\) and \(P(\mathcal{B}_i)\) in~(\ref{eq:ar_bound}). Building on these results, we obtain a new BLER formulation for both a genie-aided SC decoder and a practical SC decoder. In the genie-aided case, the decoder is assumed to have access to the
true transmitted bits when computing the LLR in~\eqref{eq:llr_out}, i.e., \(U_1^{i-1} = u_1^{i-1}\) is used to evaluate \(L_i\). In contrast, the practical SC decoder computes \(L_i\) using its own previous decisions, e.i., \(U_1^{i-1} = \hat u_1^{i-1}\).
To derive this formulation, for each \( i \in \mathcal{I} \), we decompose the error event \(\mathcal{E}_i\) into disjoint events \(\mathcal{E}_{u_1^i}\), defined for each \( u_1^i \in \mathcal{X}^i \), as
\begin{align*}
   \resizebox{0.45\textwidth}{!}{$
   \begin{aligned}
       \mathcal{E}_{u_1^i} =
   \Big\{ \big((u_1^i, v_{i+1}^N), y_1^N\big) \in
   \{u_1^i\}&\times\mathcal{X}^{N-i-1} \times \mathcal{Y}^N :\\
   &\hat{u}_1^{i-1} = u_1^{i-1},\; \hat{u}_i \neq u_i
   \Big\}.
   \end{aligned}
   $}
\end{align*}
The event \( \mathcal{E}_{u_1^i} \) corresponds to all input–output pairs for which the decoder, given the received vector \( y_1^N \) associated with the transmitted message \( (u_1^i, v_{i+1}^N) \), correctly decodes the first \( i-1 \) bits and makes its first decoding error at bit position \( i \). The decoding outcomes of the remaining bits, from positions \( i+1 \) to \( N \), are unconstrained.
Since the events \( \{\mathcal{E}_{u_1^i}\}_{u_1^i \in \mathcal{X}^i} \) are mutually disjoint, the error event \( \mathcal{E}_i \) admits the decomposition
\[
\mathcal{E}_i = \bigsqcup_{u_1^i \in \mathcal{X}^i} \mathcal{E}_{u_1^i},
\]
which implies that the corresponding error probability can be expressed as
\begin{align}\label{eq:pu_i}
    P(\mathcal{E}_i) = \sum_{u_1^i \in \mathcal{X}^i} \frac{1}{2^i}
    P\big( \mathcal{E}_{u_1^i} \big),
\end{align}
where each term is given by
\begin{align}
    \label{eq:pe_ui}
&P\big( \mathcal{E}_{u_1^i} \big) \notag\\
&= \sum_{y_1^N \in \mathcal{Y}^N}
\Pr\!\left(U_i \neq u_i,\,
U_1^{i-1} = u_1^{i-1},\, Y_1^N=y_1^N \right) \notag \\
&= \sum_{y_1^N \in \mathcal{Y}^N} \Pr(Y_1^N=y_1^N)\, \notag\\ & \cdot \Pr\!\left(U_i \neq u_i,\,
U_1^{i-1} = u_1^{i-1} \mid Y_1^N=y_1^N \right) \notag \\
&\stackrel{(a)}{=} \sum_{y_1^N \in \mathcal{Y}^N} \Pr(Y_1^N=y_1^N)\,
\Pr\!\left(U_1^{i-1} = u_1^{i-1} \mid Y_1^N=y_1^N \right) \notag \\
&\cdot
\Pr\!\left(U_i \neq u_i \mid
Y_1^N=y_1^N,\, U_1^{i-1} = u_1^{i-1} \right) \notag \\
&\stackrel{(b)}{=} \sum_{y_1^N \in \mathcal{Y}^N}
\Pr(Y_1^N=y_1^N)\, \mathcal{P}_e(u_1^i \mid y_1^N),
\end{align}
where equalities \((a)\) and \((b)\) follow from the chain rule, and
\[
\mathcal{P}_e(u_1^i \mid y_1^N)
\triangleq p_e^{(i)} \prod_{\substack{j \in \mathcal{I} \\ j < i}}
\bigl(1 - p_e^{(j)}\bigr),
\]
with \cite{balatsoukas2015llr,chandesris2018dynamic}
\begin{align} \label{eq:genie_aided}
p_e^{(j)}
&= \Pr\!\left(U_j \neq u_j\mid
Y_1^N=y_1^N,\, U_1^{j-1} = u_1^{j-1}\right) \notag\\
&= \frac{1}{1 + \exp\!\left(\lvert L_j \rvert\right)},
\end{align}
as shown in \cite{balatsoukas2015llr}.

Lemmas~\ref{prop:sym_m} and~\ref{prop:ind_u} establish two fundamental properties of
\( \mathcal{P}_e(u_1^i \mid y_1^N) \) for genie-aided SC decoding over symmetric
B-DMCs. Specifically, they show that this error probability exhibits a symmetry
with respect to input translations and does not depend on the particular
transmitted message. As a result, the analysis can be carried out without loss
of generality by assuming transmission of the all-zero sequence, which greatly
simplifies the derivation of a tractable BLER expression.

\begin{lemma}\label{prop:sym_m}
For a polar code under genie-aided SC decoding over a symmetric B-DMC
\( \mathcal{W} \), the error probability
\( \mathcal{P}_e\!\left(u_1^i \mid y_1^N \right) \) satisfies the following symmetry
property:
\[
\mathcal{P}_e\!\left(u_1^i \mid y_1^N \right)
=
\mathcal{P}_e\!\left(u_1^i \oplus a_1^i \mid b_1^N \cdot y_1^N \right),
\]
for all \( 1 \le i \le N \),
\( (u_1^i, y_1^N) \in \mathcal{X}^i \times \mathcal{Y}^N \), and
\( a_1^N \in \mathcal{X}^N \), where \( b_1^N = a_1^N \mathbf{G}_N \).
\end{lemma}

\begin{proof}
The result follows directly from the symmetry of the channel and the linearity
of the polar transform. The proof mirrors the arguments in
Proposition~13 of~\cite{arikan2009channel} and Lemma~1 of~\cite{wu2024new}, and is
therefore omitted for brevity.
\end{proof}

\begin{lemma}\label{prop:ind_u}
The event probability \( P(\mathcal{E}_{u_1^i}) \) for a polar code under a genie-aided SC decoder over a symmetric B-DMC \( \mathcal{W} \) is independent of the
transmitted input sequence and satisfies
\[
P(\mathcal{E}_{u_1^i}) = P(\mathcal{E}_{0_1^i}),
\]
where \( 0_1^i \) denotes the all-zero vector of length \( i \).
\end{lemma}

\begin{proof}
The proof follows arguments similar to Corollary~1 in~\cite{arikan2009channel}
and~\cite{wu2024new}. Let
\( x_1^N = (u_1^i, v_{i+1}^N)\mathbf{G}_N \) denote the transmitted codeword. Then,
\begin{align*}
P\!\left( \mathcal{E}_{u_1^i} \right)
&= \sum_{y_1^N \in \mathcal{Y}^N}
\Pr(Y_1^N=y_1^N)\, \mathcal{P}_e(u_1^i \mid y_1^N) \\
&= \sum_{y_1^N \in \mathcal{Y}^N}
\Pr(Y_1^N=y_1^N)\, \mathcal{P}_e(0_1^i \mid x_1^N \cdot y_1^N) \\
&= P\!\left( \mathcal{E}_{0_1^i} \right),
\end{align*}
where the last equality follows from the symmetry of \( \mathcal{W} \), which
implies
\( \{ x_1^N \cdot y_1^N : y_1^N \in \mathcal{Y}^N \} = \mathcal{Y}^N \)
for any fixed \( x_1^N \).
\end{proof}

Since \( P(\mathcal{E}_{u_1^i}) \), and hence \( \mathcal{P}_e(u_1^i \mid y_1^N) \), is independent of the transmitted message,
we introduce the simplified notation \( \mathcal{P}_{e,i}(y_1^N) \triangleq \mathcal{P}_e(u_1^i \mid y_1^N) \).
This observation leads directly to the following corollary, which provides a new expression for the BLER of a genie-aided SC decoder that assumes access to the true transmitted message.

\begin{corollary}\label{cor:bler_new}
The BLER of an \( (N,K,\mathcal{I}) \) polar code under a genie-aided SC decoder over any symmetric B-DMC \( \mathcal{W} \) can be expressed as
\begin{align*}\label{eq:our_genei}
P(\mathcal{E}) = \sum_{i \in \mathcal{I}} \sum_{y_1^N \in \mathcal{Y}^N} \Pr(Y_1^N=y_1^N)\, \mathcal{P}_{e,i}(y_1^N).
\end{align*}
\end{corollary}

\begin{proof}
For an $(N,K,\mathcal{I})$ polar code transmitted over a symmetric B-DMC $\mathcal{W}$, combining~\eqref{eq:pu_i} with Lemmas~\ref{prop:sym_m} and~\ref{prop:ind_u} yields
\begin{align*}
P(\mathcal{E}_i)
&= \sum_{u_1^i \in \mathcal{X}^i} \frac{1}{2^i}
   \sum_{y_1^N \in \mathcal{Y}^N} \Pr(Y_1^N=y_1^N)\, \mathcal{P}_{e,i}(y_1^N) \\[2pt]
&= \sum_{y_1^N \in \mathcal{Y}^N} \Pr(Y_1^N=y_1^N)\, \mathcal{P}_{e,i}(y_1^N).
\end{align*}
Summing $P(\mathcal{E}_i)$ over all $i \in \mathcal{I}$ yields the desired result.
\end{proof}

Proposition~\ref{cor:bler_new} provides a new formulation for the BLER of a \((N,K, \mathcal{I})\) polar code under genie-aided SC decoding over a symmetric B-DMC, corresponding to the first summation term in~\eqref{eq:ar_bound}. In this setting, the decoder is assumed to have access to the true transmitted message when computing the LLRs in~\eqref{eq:llr_out}. In practice, however, SC decoding relies on the decoder’s own previous decisions when evaluating~\eqref{eq:llr_out}, and the resulting BLER corresponds to the second summation term in~\eqref{eq:ar_bound}. Therefore, to derive a BLER formulation of a \((N,K, \mathcal{I})\) polar code under a practical SC decoder, we must develop new expressions for the error probabilities \(P(\mathcal{B}_i)\) for each \(i \in \mathcal{I}\). To this end, for each \( i \in \mathcal{I} \), and in analogy with the decomposition of \( \mathcal{E}_i \), we decompose the event \( \mathcal{B}_i \) into disjoint events \( \mathcal{B}_{u_1^i} \), defined for each \( u_1^i \in \mathcal{X}^i \), as
\begin{align*}
\resizebox{0.48\textwidth}{!}{$
\mathcal{B}_{u_1^i}
=
\Big\{
\big((u_1^i, v_{i+1}^N), y_1^N\big)
\in
\{u_1^i\} \times \mathcal{X}^{N-i-1} \times \mathcal{Y}^N
:\,
u_i \neq \hat u_i
\Big\}.
$}
\end{align*}

That is, \( \mathcal{B}_{u_1^i} \) consists of all input–output pairs for which the
decoder, given the received vector \( y_1^N \) corresponding to the transmitted
message \( (u_1^i, v_{i+1}^N) \), incorrectly decodes the \( i \)-th bit, regardless
of the correctness of the other bits. Since the events
\( \{\mathcal{B}_{u_1^i}\}_{u_1^i \in \mathcal{X}^i} \) are mutually disjoint, the
error event \( \mathcal{B}_i \) admits the decomposition
\[
\mathcal{B}_i = \bigsqcup_{u_1^i \in \mathcal{X}^i} \mathcal{B}_{u_1^i},
\]
which implies that the corresponding error probability can be expressed as
\begin{align}\label{eq:p_sc}
P(\mathcal{B}_i)
= \sum_{u_1^i \in \mathcal{X}^i} \frac{1}{2^i}
P\!\left(\mathcal{B}_{u_1^i}\right),
\end{align}
where
\begin{align*}
P\!\left(\mathcal{B}_{u_1^i}\right)
= \sum_{y_1^N \in \mathcal{Y}^N}
\Pr(Y_1^N=y_1^N)\, \mathcal{P}_e'(u_1^i \mid y_1^N),
\end{align*}
and
\begin{align*}
\mathcal{P}_e'(u_1^i \mid y_1^N) \triangleq
\Pr\!\left(
U_i \neq \hat u_i , \; U_1^{i-1} = \hat{u}_1^{i-1} \mid Y_1^N=y_1^N
\right).
\end{align*}
The quantity \( \mathcal{P}_e'(u_1^i \mid y_1^N) \) represents the probability that the \( i \)-th bit is decoded incorrectly, namely \(\hat{u}_i\oplus 1\) is correct, while the previous decoder decisions are \( \hat{u}_1^{i-1} \), given the received vector \( y_1^N \). By the chain rule, this probability can be factorized as
\begin{align}\label{eq:p_prime}
&\mathcal{P}_e'(u_1^i \mid y_1^N) \notag \\ 
&=
\Pr\!\left(
U_i \neq \hat u_i \mid Y_1^N=y_1^N,\; U_1^{i-1} =\hat{u}_1^{i-1}
\right) \quad \quad\\ 
&\cdot
\Pr\!\left(
U_1^{i-1}=\hat{u}_1^{i-1} \mid Y_1^N=y_1^N
\right) \notag.
\end{align}
Applying the chain rule again, the second term can be written as
\begin{align*}
    \Pr\!\big(U_1^{i-1} &=\hat{u}_1^{i-1} \mid Y_1^N= y_1^N \big)=\\& \prod_{j=1}^{i-1}
\Pr\!\left(U_j=\hat{u}_j \mid Y_1^N= y_1^N,\; U_1^{j-1} = \hat{u}_1^{j-1}\right).
\end{align*}
For frozen indices \( j \in \mathcal{I}^c \), we have
\( \Pr(U_j =\hat{u}_j \mid Y= y_1^N,U_1^{j-1} =\hat{u}_1^{j-1}) = 1 \). For information indices
\( j \in \mathcal{I} \), it follows from~\cite{balatsoukas2015llr} that
\[
\Pr\!\left(U_j = \hat{u}_j \mid Y_1^N= y_1^N,\; U_1^{j-1} =\hat{u}_1^{j-1}\right)
= \frac{1}{1 + \exp\!\left(-|L_j|\right)},
\]
where \( L_j \) denotes the decision LLR corresponding to bit \( j \). Therefore,
\begin{align*}
\Pr\!\left(U_1^{i-1}=\hat{u}_1^{i-1} \mid Y_1^N = y_1^N \right)
&= \prod_{\substack{j \in \mathcal{I} \\ j < i}}
\frac{1}{1 + \exp\!\left(-|L_j|\right)}.
\end{align*}
Similarly, the first term in~\eqref{eq:p_prime} can be approximated
as~\cite{balatsoukas2015llr}
\begin{align}\label{eq:sc_prac}
\Pr\!\left(U_i \neq u_i \mid Y_1^N= y_1^N,\; U_1^{i-1} =\hat{u}_1^{i-1}\right)
= \frac{1}{1 + \exp\!\left(|L_i|\right)}.
\end{align}
Furthermore, the symmetry and input-independence properties established for \( \mathcal{P}_e(u_1^i \mid y_1^N) \) in Lemmas~\ref{prop:sym_m} and~\ref{prop:ind_u} also hold for \( \mathcal{P}_e'(u_1^i \mid y_1^N) \). Consequently, we omit the
explicit dependence on \( u_1^i \) and define \( \mathcal{P}_{e,i}'(y_1^N) \triangleq \mathcal{P}_e'(u_1^i \mid y_1^N) \), which leads to the following result.

\begin{proposition}\label{pro:fer_form} 
The BLER of an \( \left( N, K, \mathcal{I} \right) \) polar code under SC decoding
    over a symmetric B-DMC \( \mathcal{W} \) satisfies  
    \begin{align}\label{eq:fer_bound}
     \resizebox{0.43\textwidth}{!}{$
     \begin{aligned}
         P(\mathcal{E}) &= \sum_{i\in \mathcal{I}} P(\mathcal{E}_i)
                    =  \sum_{i\in \mathcal{I}} \sum_{y_1^N\in\mathcal{Y}^N} \Pr(Y_1^N=y_1^N) \mathcal{P}_{e,i}\left(y_1^N\right),\\
                   &\leq \sum_{i\in \mathcal{I}} P(\mathcal{B}_i)
                   = \sum_{i\in \mathcal{I}} \sum_{y_1^N\in\mathcal{Y}^N} \Pr(Y_1^N=y_1^N) \mathcal{P}_{e,i}^\prime\left(y_1^N\right).
     \end{aligned}
     $}
    \end{align}
\end{proposition}

\begin{proof}
    The proof follows directly from (\ref{eq:ar_bound}), (\ref{eq:p_sc}), and Corollary~\ref{cor:bler_new}.
\end{proof}

\begin{figure*}[t]
    \centering
    \begin{minipage}{0.32\textwidth}
    \centering
        \includegraphics[width=0.8\linewidth]{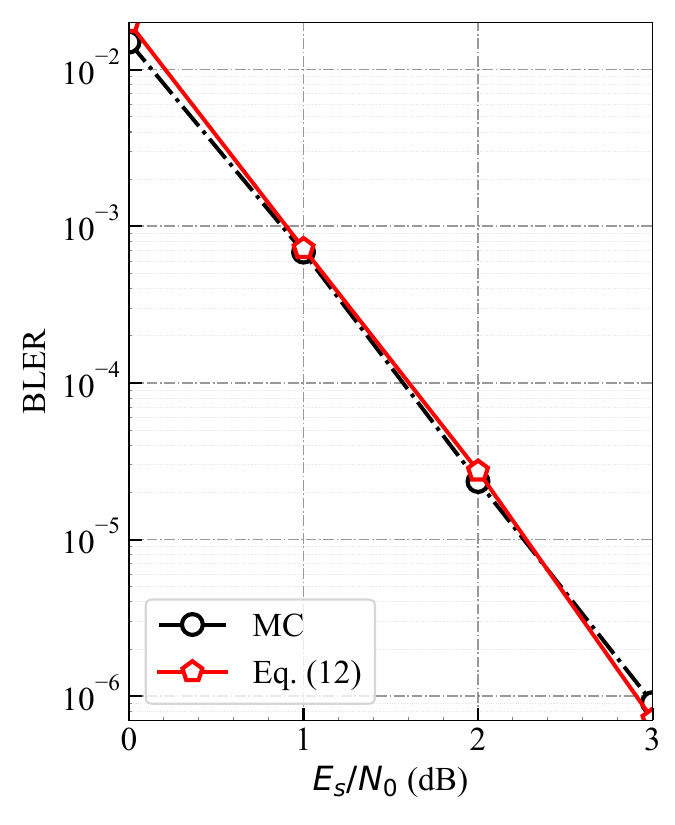}
        \caption*{(a) AWGN}
    \end{minipage}
    \begin{minipage}{0.32\textwidth}
    \centering
        \includegraphics[width=0.8\linewidth]{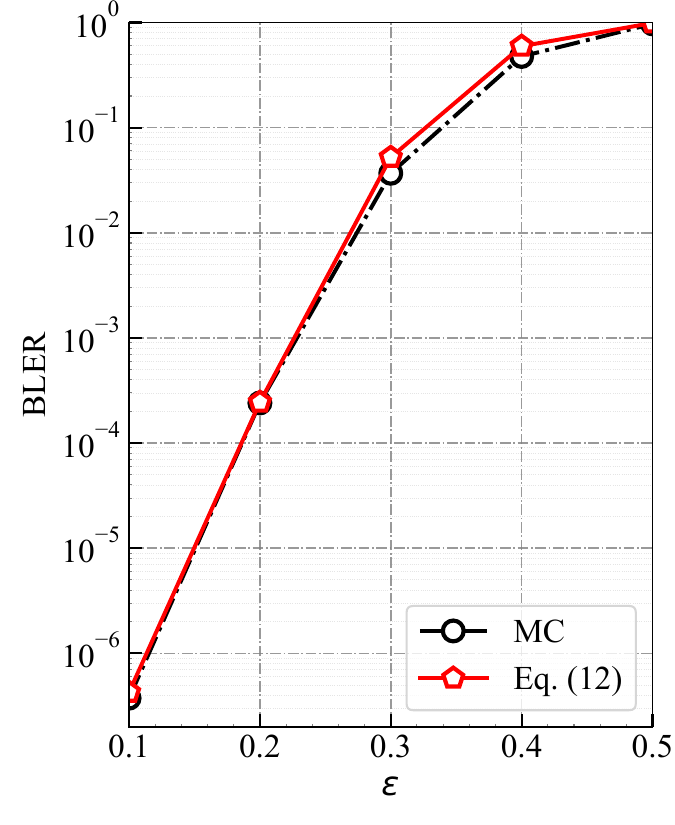}
        \caption*{(b) BEC}
    \end{minipage}
    \begin{minipage}{0.32\textwidth}
    \centering
        \includegraphics[width=0.8\linewidth]{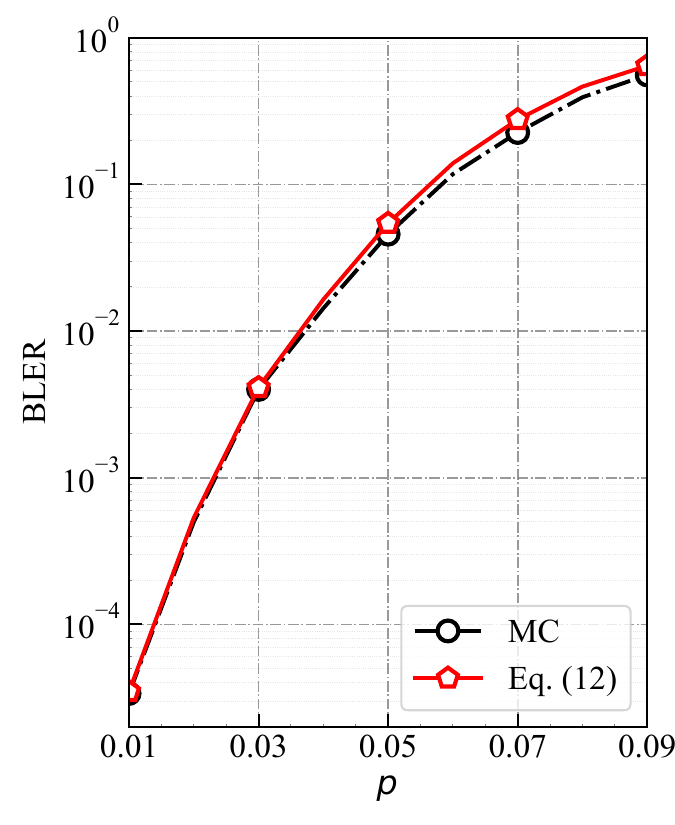}
        \caption*{(c) BSC}
    \end{minipage}
    \caption{Subfigures (a)–(c) show the BLER performance of the 5G polar code with \((N, K) = (256, 128)\) over AWGN, BEC, and BSC channels under SC decoding. The BLERs are obtained through Monte Carlo simulations and evaluated using ~(\ref{eq:fer_bound}). }
    \label{fig:fer_bound}
\end{figure*}

Figure~\ref{fig:fer_bound} shows the BLER for a 5G polar code of length \( N = 256 \) and dimension \( K = 128 \), evaluated via MC simulation and our derivation for the BLER of a given polar code under a practical SC decoding in ~(\ref{eq:fer_bound}). The all-zero codeword is transmitted over various channels using BPSK modulation, including AWGN, Binary Erasure Channel (BEC), and Binary Symmetric Channel (BSC). In the MC simulation, the process is repeated until at least 100 erroneous blocks are observed. For the proposed approximation, the BLER is computed using the same set of channel samples generated during the MC simulation. As shown in Figure~\ref{fig:fer_bound}, the BLER approximated by ~(\ref{eq:fer_bound}) closely matches the MC simulation results, confirming its accuracy as a reliable approximation of the BLER of a polar code under SC decoding over various channels, which shows ~(\ref{eq:fer_bound}) is channel-agnostic, and it depends only on SC decoder.

\begin{figure}[t]
    \centering
    \includegraphics[width=\linewidth]{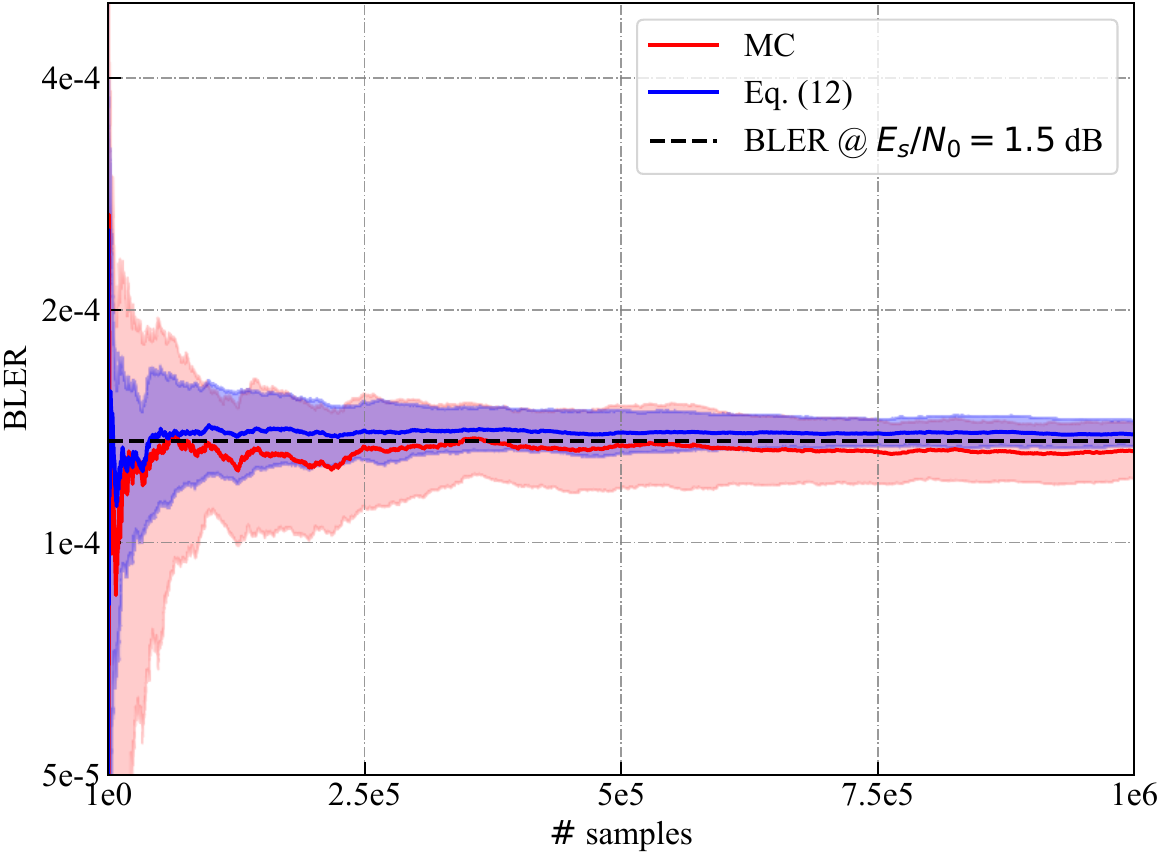}
    \caption{Convergence of BLER estimation at $E_s/N_0 = 1.5$~dB for a $(256,128)$ polar code.
MC (red) and the proposed approximation in~(\ref{eq:fer_bound}) (blue) are shown
versus the number of channel samples. Solid lines indicate the mean over
multiple runs, shaded regions denote $\pm$ one standard deviation, and the
dashed line represents a high-precision MC reference.
}
    \label{fig:r_err}
\end{figure}

This generality enables the proposed formulation to be applied to a wide range
of channels, with a key application being the construction of polar codes for
arbitrary symmetric B-DMCs (see Section~\ref{sec:SC_cons}). The proposed
algorithm follows an MC-like, channel-agnostic approach; however, unlike
conventional MC methods that rely solely on observed error
events, it exploits decoder soft information to estimate both bit-wise and
block error probabilities, thereby improving sample efficiency.

To illustrate this advantage, we compare the convergence behavior of the
proposed BLER estimator with that of the conventional MC simulation at
$E_s/N_0 = 1.5$~dB as a function of the number of channel samples.
Fig.~\ref{fig:r_err} shows the MC BLER estimate (red) and the proposed
semi-analytical approximation in~(\ref{eq:fer_bound}) (blue). Solid lines denote
the mean BLER obtained over multiple independent runs, while the shaded
regions indicate $\pm$ one standard deviation, capturing the estimation
variability. The dashed horizontal line corresponds to a high-precision MC
reference BLER obtained by simulating until at least $500$ frame errors are
observed. Although both methods converge to the same BLER as the number of
samples increases, the proposed approximation exhibits reduced variance and
faster convergence, demonstrating its improved sample efficiency relative to
standard MC simulation.

As another application of~\eqref{eq:fer_bound}, the proposed formulation can be used to reduce the computational complexity of DSCF decoding by limiting the set of candidate bit-flipping positions. Specifically, only bit indices satisfying $\mathbb{E}\!\left[\mathcal{P}_{e,i}'\right] \ge \gamma$ are considered, where
\begin{align}\label{eq:bit_errors}
\mathbb{E}\!\left[\mathcal{P}_{e,i}'\right]
= \sum_{y_1^N \in \mathcal{Y}^N} \Pr(Y_1^N=y_1^N)\, \mathcal{P}_{e,i}'(y_1^N),
\end{align}
denotes the error probability of bit $i \in \mathcal{I}$ and $\gamma \in [0,1]$ is a threshold parameter.
To evaluate this approach, we consider a $(256,128)$ 5G polar code concatenated with a 16-bit CRC and decoded using DSCF with bit-flip order $\omega=1$ and at most $T=10$ additional attempts. The quantities $\mathbb{E}\!\left[\mathcal{P}_{e,i}'\right]$ are approximated at $E_s/N_0=-2$~dB, and bit-flipping is restricted to positions satisfying the threshold condition with $\gamma \in \{0.001, 0.003, 0.005\}$.
For $\gamma=0.001$, the number of candidate bits is reduced from $144$ to at most $40$, corresponding to a $60\%$ reduction in computational effort, without noticeable performance degradation, as shown in Fig.~\ref{fig:fer_dscf_mdscf}. Higher thresholds yield further complexity reductions at the cost of degraded error-correction performance.

\begin{figure}[t]
    \centering
    \begin{minipage}{0.4\textwidth}
    \centering
        \includegraphics[width=0.8\linewidth]{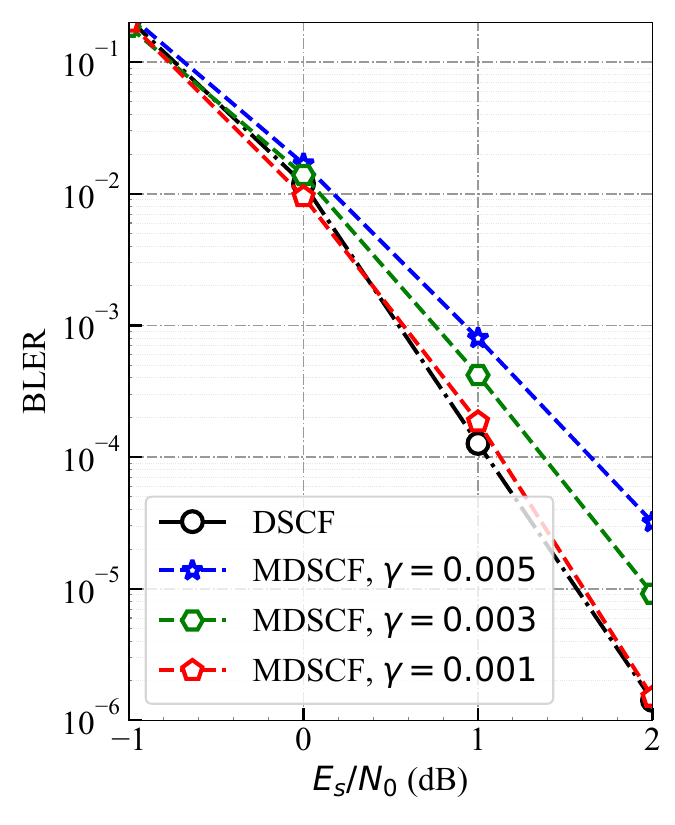}
    \end{minipage}
    \caption{The BLER of the a $(256,128)$ polar code under both conventional and modified DSCF (MDSCF) decoding, evaluated with three different thresholds \(\gamma \in \{0.001, 0.003, 0.005\}\).}
    \label{fig:fer_dscf_mdscf}
\end{figure}

\section{SC-Optimized Polar Code Construction}\label{sec:SC_cons}

In this section, we illustrate how the proposed block error probability formulation can be employed to construct polar
codes adapted to a specific channel.
To design a polar code based on ~(\ref{eq:fer_bound}), consider \(\mathbb{E}\left[ \mathcal{P}_{e,i}^\prime \right]\) in ~(\ref{eq:bit_errors}) denoting the error probability for bit index \( i \in \mathcal{I} \).  According to Proposition~\ref{pro:fer_form}, the BLER of a \((N, K)\) polar code, with information set \( \mathcal{I} \), is upper bounded by
\begin{align}\label{eq:fer}
    P(\mathcal{E}) \leq \sum_{i \in \mathcal{I}} \mathbb{E}\left[ \mathcal{P}_{e,i}^\prime \right] 
                   \leq \sum_{i \in \mathcal{I}} Z\left(\mathcal{W}_N^{(i)}\right).
\end{align}

Based on ~\eqref{eq:fer}, we propose an incremental and nested-like construction method. This algorithm begins with a polar code of rate one, and it progressively identifies and freezes the least reliable bit-channels by estimating the error probability \(\mathbb{E}[\mathcal{P}_{e,i}^\prime]\) for each \(i \in \mathcal{I}\) under SC decoding.

The proposed construction method begins with all bit indices considered as candidate information bits, i.e., 
\(\mathcal{I}_0 \gets \{1, 2, \dots, N\}\), yielding an initial code dimension of \(K_0 \gets N\), 
and an empty frozen set \(\mathcal{I}_0^c \gets \emptyset\).
At each iteration \(t\), the algorithm computes the error metric 
\(\mathbb{E}[\mathcal{P}_{e,i}^\prime]\) for all \(i \in \mathcal{I}_t\), and identifies the index \(i^*\) 
with the highest value—indicating the least reliable (most error-prone) bit. This index is then moved to the frozen set.

\begin{algorithm}[t]
\caption{SC-Optimized Polar Code Construction}
\label{alg:incremental_polar}
\begin{algorithmic}[1]
\State \textbf{Input:} $N$, $K$, channel $\mathcal{W}$, target $E_s/N_0$
\State \textbf{Output:} Information set $\mathcal{I}$
\State Initialize $\mathcal{I}_0 \gets \{1,\dots,N\}$, $K_0 \gets N$, $t \gets 0$
\While{$K_t > K$}
    \State $\{\mathbb{E}[\mathcal{P}_{e,i}']\}_{i \in \mathcal{I}_t}
    \gets \text{Estimator}(N,\mathcal{I}_t,\mathcal{W},E_s/N_0)$
    \State $i^\ast \gets \arg\max_{i \in \mathcal{I}_t} \mathbb{E}[\mathcal{P}_{e,i}']$
    \State $\mathcal{I}_{t+1} \gets \mathcal{I}_t \setminus \{i^\ast\}$
    \State $K_{t+1} \gets |\mathcal{I}_{t+1}|$, \quad $t \gets t+1$
\EndWhile
\State \Return $\mathcal{I} \gets \mathcal{I}_t$
\end{algorithmic}
\end{algorithm}

\begin{algorithm}[t]
\caption{Estimator for $\mathbb{E}[\mathcal{P}_{e,i}']$}
\label{alg:estimator}
\label{proc:estimator}
\begin{algorithmic}[1]
\Procedure{Estimator}{$N,\mathcal{I},\mathcal{W},E_s/N_0$}
    \State $\mathcal{T}_{\mathrm{err}}\gets0$, $n\gets0$, $s_i\gets0\;\forall i\in\mathcal{I}$
    \While{$\mathcal{T}_{\mathrm{err}}<\mathcal{T}_{\mathrm{target}}$}
        \State $n\gets n+1$
        \State Generate $y_1^N$ via $\mathcal{W}$ and decode using SC decoding
        \ForAll{$i\in\mathcal{I}$}
            \State $s_i \gets s_i + \mathcal{P}_{e,i}'(y_1^N)$
        \EndFor
        \If{a decoding error occurs}
            \State $\mathcal{T}_{\mathrm{err}}\gets\mathcal{T}_{\mathrm{err}}+1$
        \EndIf
    \EndWhile
    \State \Return $\{s_i/n\}_{i\in\mathcal{I}}$
\EndProcedure
\end{algorithmic}
\end{algorithm}

At iteration $t$, the expectation $\mathbb{E}[\mathcal{P}_{e,i}']$ is approximated using $M$ independent channel output samples $y_{1,d}^N \in \mathcal{Y}^N$, $d=1,\dots,M$, generated by simulating transmissions over $\mathcal{W}$ with a polar code of length $N$, dimension $K_t$, and information set $\mathcal{I}_t$. Since~\eqref{eq:fer} is independent of the transmitted message, an all-zero codeword is assumed. Each sample is decoded using SC decoding, and the resulting LLRs are used to compute $\mathcal{P}_{e,i}'(y_{1,d}^N)$ according to~\eqref{eq:p_prime}. The expectation is then estimated by the empirical mean
\begin{align}\label{eq:er_appr}
\mathbb{E}[\mathcal{P}_{e,i}'] \approx \frac{1}{M} \sum_{d=1}^{M} \mathcal{P}_{e,i}'\!\left(y_{1,d}^N\right).
\end{align}
Sampling continues until convergence, which is determined by observing a predefined number of decoding errors $\mathcal{T}_{\mathrm{target}}$. The approximation procedure is summarized in Algorithm~\ref{alg:estimator}.

After identifying the least reliable bit-channel, the information set is updated as $\mathcal{I}_{t+1} := \{1,\dots,N\}\setminus\mathcal{I}_{t+1}^c$, with code dimension $K_{t+1}=|\mathcal{I}_{t+1}|$. This process is repeated until the target dimension $K$ is reached. The complete construction procedure is summarized in Algorithms~\ref{alg:incremental_polar} and~\ref{alg:estimator}.

\section{Block error rate Under SCL decoding}\label{sec:scl_bep}

In this section, we analyze the BLER of polar codes under SCL decoding over a symmetric B-DMC $\mathcal{W}$. The main objective is to decompose the SCL BLER into interpretable error components and to derive a semi-analytical expression that quantifies the contribution of each information bit to the overall decoding failure. This formulation reveals how individual bit decisions affect path survival during list decoding and provides a principled basis for constructing polar codes optimized for SCL decoding with a given list size.

To this end, consider a \((N,K,\mathcal{I})\) polar code transmitted over a symmetric B-DMC $\mathcal{W}:\mathcal{X}\to\mathcal{Y}$. Let SCL decoding with list size $L=2^m$ be employed, where $m\in [\![1,+\infty]\!]$ (the restriction to powers of two is assumed for simplicity). Given a received sequence $y_1^N\in\mathcal{Y}^N$ corresponding to the transmission of $u_1^N\in\mathcal{X}^N$, the SCL decoder produces a list of $L$ candidate paths 
$\{\hat{u}[l]\}_{l=1}^L$ and selects the final output as
\[
l^*=\arg\min_{l\in[\![1,L]\!]}\mathrm{PM}[l]_N,
\]
where $\mathrm{PM}[l]_N$ denotes the path metric of the $l$-th path, defined in~\eqref{eq:path_metr}. 


Similar to Section~\ref{sec:bep}, for SCL decoding with list size \(L\), we define the random vectors 
\((U_1^N, X_1^N, Y_1^N, \hat{U}[1]_1^N, \ldots, \hat{U}[L]_1^N)\) on the probability space 
\((\mathcal{X}^N \times \mathcal{Y}^N, \mathbb{P})\)~\cite{arikan2009channel}. 
Here, \(U_1^N\) denotes the encoder input, \(X_1^N\) the channel input, 
\(Y_1^N\) the channel output, and \(\hat{U}[l]_1^N\), for 
\(l \in [\![1,L]\!]\), the \(l\)-th decoding path produced by the SCL decoder.  For each realization \((u_1^N, y_1^N) \in \mathcal{X}^N \times \mathcal{Y}^N\), 
the first three vectors take the values \(U_1^N = u_1^N\), 
\(X_1^N = u_1^N \mathbf{G}_N\), and \(Y_1^N = y_1^N\), 
while the decoder outputs \(\hat{U}[l]_1^N\), 
\(l \in [\![1,L]\!]\), are generated according to the SCL decoding rule.

The block error event of a polar code under SCL decoding is then defined as
\begin{align}
\mathcal{E}_{\mathrm{SCL}}
=\Big\{(u_1^N,y_1^N)\in\mathcal{X}^N\times\mathcal{Y}^N:
u_{\mathcal{I}}\neq \hat u[l^*]_{\mathcal{I}}\Big\}.
\end{align}

An SCL decoding error can occur through two distinct mechanisms \cite{piao2023performance}. First, the correct decoding path may be eliminated at some decoding stage due to list pruning, resulting in a \emph{path loss} event. Second, the correct path may survive until the end of decoding but not be selected as the final output, leading to a \emph{path selection} error.
Accordingly, the BLER under SCL decoding can be decomposed as 
\begin{equation}
\label{eq:scl_err_prob}
P(\mathcal{E}_{\mathrm{SCL}})
= P(\mathcal{E}_{\mathrm{PL}}) + P(\mathcal{E}_{\mathrm{PS}}),
\end{equation}
where $\mathcal{E}_{\mathrm{PL}}$ denotes the path loss event and $\mathcal{E}_{\mathrm{PS}}$ denotes the path selection event.

As shown in~\cite{tal2015list,niu2012crc}, CRC-aided SCL decoding can significantly reduce $P(\mathcal{E}_{\mathrm{PS}})$ by assisting the selection of the correct codeword from the surviving list. In contrast, the path loss probability $P(\mathcal{E}_{\mathrm{PL}})$ is fundamentally governed by the structure of the polar code, in particular, the choice of information and frozen bit indices \cite{choi2025sparsely}. Therefore, improving the BLER performance of polar codes under SCL decoding requires reducing both components, with code construction playing a critical role in mitigating path loss.

In the following subsection, we focus on a detailed analysis of $P(\mathcal{E}_{\mathrm{PL}})$ and derive a novel expression that quantifies the contribution of each information bit to the path loss probability and, consequently, to the overall SCL BLER.

\subsection{Path Loss Error Probability}
Consider SCL decoding with list size $L=2^m$ ($m\in[\![1,+\infty]\!]$) applied to
a received vector $y_1^N\in\mathcal{Y}^N$ corresponding to the transmission of
$u_1^N\in\mathcal{X}^N$ over a symmetric B-DMC $\mathcal{W}:\mathcal{X}\to\mathcal{Y}$.
At decoding stage $i$, let
\[
\mathcal{V}_i \triangleq \big\{\hat{u}[l]_1^i : l\in[\![1,L]\!]\big\}
\]
denote the set of length-$i$ prefixes of the $L$ surviving paths. The path loss event is defined as the event that the transmitted codeword is not contained in the final candidate list, i.e.,
\begin{align}
\label{eq:def_epl}
\mathcal{E}_{\mathrm{PL}}
\triangleq
\Big\{(u_1^N,y_1^N)\in\mathcal{X}^N\times\mathcal{Y}^N:
u_1^N\neq \hat u[l]_1^N,\ \forall\, l\in[\![1,L]\!]\Big\}.
\end{align}
Since no pruning occurs before the list reaches size $L$, the correct decoding path can only be eliminated starting from the $m$-th information bit. Let $\mathcal{I}=\{i_1<i_2<\cdots<i_K\}$ denote the ordered information set, and let $i_m$ be the $m$-th information index. For each $i\in\mathcal{I}$ with $i\ge i_m$, define the stage-$i$ path-loss event as
\begin{align}
\label{eq:def_ei_scl}
\mathcal{E}^{\mathrm{SCL}}_i
\triangleq
\Big\{(u_1^N,y_1^N)\in\mathcal{X}^N\times\mathcal{Y}^N:
\mathcal{S}_{i-1},\; u_1^{i}\notin\mathcal{V}_{i}\Big\},
\end{align}
where
\[
\mathcal{S}_{i-1} \triangleq \left\{u_1^j \in \mathcal{V}_j;\; j\in [\![1,i-1]\!]\right\}
\]
denotes the event that the correct decoding path remains among the surviving paths at all decoding stages up to bit $i-1$. The event $\mathcal{E}^{\mathrm{SCL}}_i$ therefore corresponds to the situation in which the correct path survives all pruning operations up to stage $i-1$ but is eliminated at stage $i$ due to list pruning. The events $\{\mathcal{E}^{\mathrm{SCL}}_i\}_{i\in\mathcal{I},\,i\ge i_m}$ are mutually disjoint and their union equals $\mathcal{E}_{\mathrm{PL}}$, hence
\begin{align}
\label{eq:epl_decomp}
\mathcal{E}_{\mathrm{PL}}
=
\bigsqcup_{\substack{i\in\mathcal{I}\\ i\ge i_m}}
\mathcal{E}^{\mathrm{SCL}}_i,
\qquad
P(\mathcal{E}_{\mathrm{PL}})
=
\sum_{\substack{i\in\mathcal{I}\\ i\ge i_m}}
P\!\left(\mathcal{E}^{\mathrm{SCL}}_i\right).
\end{align}
For each such $i$, the probability of the stage-$i$ path-loss event can be
written as
\begin{align}
\label{eq:pei_scl_expect}
P\!\left(\mathcal{E}^{\mathrm{SCL}}_i\right)&=
\sum_{y_1^N\in\mathcal{Y}^N}
\Pr(Y_1^N = y_1^N) \notag\\
& \cdot \sum_{u_1^N\in\mathcal{X}^N}
\frac{1}{2^N}
\Pr\!\left(u_1^{i}\notin\mathcal{V}_{i},\,\mathcal{S}_{i-1}\mid Y_1^N = y_1^N\right) \notag\\
&\stackrel{(a)}{=}
\sum_{y_1^N\in\mathcal{Y}^N}
\Pr(Y_1^N = y_1^N)\,
\mathcal{P}^{\mathrm{SCL}}_{e,i}(y_1^N),
\end{align}
where equality $(a)$ follows from the fact that $\Pr\!\left(u_1^{i}\notin\mathcal{V}_{i},\,\mathcal{S}_{i-1}\mid Y_1^N = y_1^N\right)$ is independent of the transmitted sequence $u_1^N$ for symmetric B-DMCs. This property can be established using arguments analogous to those employed to show the message-independence of the SC bit-error probability $\mathcal{P}_e(u_1^i\mid  y_1^N)$.
For notational convenience, in (\ref{eq:pei_scl_expect}), we define \(\mathcal{P}^{\mathrm{SCL}}_{e,i}(y_1^N)
\triangleq
\Pr\!\left(u_1^{i}\notin\mathcal{V}_{i},\,\mathcal{S}_{i-1}\mid Y_1^N = y_1^N\right),\)
which represents the conditional probability that the correct decoding path is eliminated at stage $i$, given the observation $y_1^N$.

Following~\cite{afisiadis2014low,shen2021dynamic}, $\mathcal{P}^{\mathrm{SCL}}_{e,i}(y_1^N)$ admits the product form
\begin{align}
\label{eq:pei_scl_factor}
\mathcal{P}^{\mathrm{SCL}}_{e,i}(y_1^N)
=
\mathcal{P}^{(i)}_{e}(y_1^N)
\prod_{\substack{j\in\mathcal{I}\\ i_m\le j < i}}
\bigl(1-\mathcal{P}^{(j)}_{e}(y_1^N)\bigr),
\end{align}
where, for $j\in\mathcal{I}$,
\[
\mathcal{P}^{(j)}_{e}(y_1^N)
\triangleq
\Pr\!\left(u_1^{j}\notin\mathcal{V}_{j}\mid Y_1^N = y_1^N,\,\mathcal{S}_{j-1}\right)
\]
denotes the conditional probability that the correct path is eliminated at stage $j$, given that it has survived all previous pruning stages. This probability can be approximated as
\begin{align}
\label{eq:Pej_scl_approx}
\mathcal{P}^{(j)}_{e}(y_1^N)
\approx
\frac{1}{1+\exp(\mathbb{L}_j)},
\end{align}
where the reliability metric $\mathbb{L}_j$ is computed from the path metrics at
stage $j$ as
\begin{align}
\label{eq:scl_llr}
\mathbb{L}_j
\triangleq
\ln\!\left(
\frac{\sum_{l=1}^{L}\exp\!\big(-\mathrm{PM}[l]_j\big)}
     {\sum_{l=L+1}^{2L}\exp\!\big(-\mathrm{PM}[l]_j\big)}
\right),
\end{align}
assuming that the $2L$ candidate path metrics
$\{\mathrm{PM}[l]_j\}_{l=1}^{2L}$ are sorted such that
$l\in[\![1,L]\!]$ corresponds to the surviving paths and
$l\in[\![L+1,2L]\!]$ to the discarded paths.

Note that, in the definitions of $\mathbb{L}_j$ in~\eqref{eq:scl_llr} and $\mathcal{E}^{\mathrm{SCL}}_i$ in~\eqref{eq:def_ei_scl}, a genie-aided assumption is implicitly made, namely that the correct prefix $u_1^{j-1}$ is retained during path pruning and remains among the surviving paths from stage $1$ to stage $j-1$. In practice, however, SCL decoding has no access to the true transmitted path and retains candidate paths solely based on their path metrics, which are computed using the decoder’s own previous decisions. Therefore, to characterize the PL error probability of an $(N,K,\mathcal{I})$ polar code under practical SCL decoding, it is necessary to replace the genie-aided formulation with a tractable bound that does not condition on the survival of the correct path and depends only on the decoder’s observed list evolution. To this end, we introduce a superevent $\mathcal{B}^{\mathrm{SCL}}_i$ of $\mathcal{E}^{\mathrm{SCL}}_i$, defined as
\begin{align}
\label{eq:def_bi_scl}
\mathcal{B}^{\mathrm{SCL}}_i
\triangleq
\Big\{(u_1^N,y_1^N)\in\mathcal{X}^N\times\mathcal{Y}^N:\;
\mathcal{S}'_{i-1},\; u_1^{i}\notin\mathcal{V}_{i}\Big\},
\end{align}
where $\mathcal{S}'_{i-1}$ is obtained from $\mathcal{S}_{i-1}$ by removing the
condition that the correct path belongs to the surviving set, i.e.,
\[
\mathcal{S}'_{i-1} \triangleq \left\{\mathcal{V}_j: \; j \in [\![1,i-1]\!]\right\}.
\]
Therefore, $\mathcal{E}^{\mathrm{SCL}}_i \subseteq \mathcal{B}^{\mathrm{SCL}}_i$, and 
\begin{align}
\label{eq:epl_bound}
P(\mathcal{E}_{\mathrm{PL}})
=
\sum_{\substack{i\in\mathcal{I}\\ i\ge i_m}}
P\!\left(\mathcal{E}^{\mathrm{SCL}}_i\right)
\le
\sum_{\substack{i\in\mathcal{I}\\ i\ge i_m}}
P\!\left(\mathcal{B}^{\mathrm{SCL}}_i\right).
\end{align}
For each $i$, the probability of $\mathcal{B}^{\mathrm{SCL}}_i$ can be written as
\begin{align}
\label{eq:bi_expect}
&P\!\left(\mathcal{B}^{\mathrm{SCL}}_i\right) \notag\\
&=
\sum_{y_1^N\in\mathcal{Y}^N} 
\Pr(Y_1^N = y_1^N)\,
\Pr\!\left(u_1^i\notin\mathcal{V}_i,\,
\mathcal{S}'_{i-1}\mid Y_1^N = y_1^N\right) \notag\\
&\triangleq
\sum_{y_1^N\in\mathcal{Y}^N}
\Pr(Y_1^N = y_1^N)\,\mathcal{P}^{\prime\mathrm{SCL}}_{e,i}(y_1^N),
\end{align}
where
\begin{align*}
&\mathcal{P}^{\prime\mathrm{SCL}}_{e,i}(y_1^N) \\
&\triangleq
\Pr\!\left(u_1^i\notin\mathcal{V}_i,\,
\mathcal{S}'_{i-1}\mid Y_1^N = y_1^N\right) \\
&=
\Pr\!\left(u_1^i\notin\mathcal{V}_i \mid Y_1^N = y_1^N,\,
\mathcal{S}'_{i-1}\right)
\Pr\!\left(\mathcal{S}'_{i-1}\mid Y_1^N = y_1^N\right).
\end{align*}
The first term can be approximated by
\[
\Pr\!\left(u_1^i\notin\mathcal{V}_i \mid Y_1^N = y_1^N,\,
\mathcal{S}'_{i-1}\right)
\approx
\frac{1}{1+\exp(\mathbb{L}_i)},
\]
while the second term admits the factorization
\begin{align}
\label{eq:correct_path_trajectory}
\Pr\!\left(\mathcal{S}'_{i-1}\mid Y_1^N = y_1^N\right)
=
\prod_{j=1}^{i-1}
\Pr\!\left(\mathcal{V}_j \mid Y_1^N = y_1^N,\,
\mathcal{S}'_{j-1}\right).
\end{align}
Moreover, for each $j$,
\[
\Pr\!\left(\mathcal{V}_j \mid Y_1^N = y_1^N,\,
\mathcal{S}'_{j-1}\right)
\approx
\frac{1}{1+\exp(-\mathbb{L}_j)},
\]
where
\[
\mathbb{L}_j
=
\ln\!\left(
\frac{\sum_{l=1}^{L}\exp\!\big(-\mathrm{PM}[l]_j\big)}
     {\sum_{l=L+1}^{2L}\exp\!\big(-\mathrm{PM}[l]_j\big)}
\right)
\]
quantifies the reliability of the decoder decision at stage $j$. Larger values of $\mathbb{L}_j$ indicate a higher likelihood that the correct path remains in the surviving list. Finally, since no pruning occurs for frozen indices or for information indices $j<i_m$, we have
$\Pr(\mathcal{V}_j \mid Y_1^N = y_1^N,\mathcal{S}'_{j-1})=1$ in these cases, which yields
\[
\Pr\!\left(\mathcal{S}'_{i-1}\mid Y_1^N = y_1^N\right)
\approx
\prod_{\substack{j\in\mathcal{I}\\ i_m\le j<i}}
\frac{1}{1+\exp(-\mathbb{L}_j)}.
\]
Consequently, for a given received vector $y_1^N$, the probability that the
correct path is eliminated at decoding stage $i$ under practical SCL decoding
can be approximated as
\begin{align}\label{eq:pp_i_SCL}
    \mathcal{P}^{\prime\mathrm{SCL}}_{e,i}(y_1^N)
=
\frac{1}{1+\exp(\mathbb{L}_i)}
\prod_{\substack{j\in\mathcal{I}\\ i_m\le j<i}}
\frac{1}{1+\exp(-\mathbb{L}_j)}.
\end{align}
Based on the preceding analysis of the PL error probability for an $(N,K,\mathcal{I})$ polar code under both genie-aided and practical SCL decoding over a symmetric B-DMC $\mathcal{W}:\mathcal{X}\to\mathcal{Y}$, we now state the following proposition, which characterizes the BLER of polar codes under genie-aided and practical SCL decoding.

\begin{proposition}\label{prop:scl_blck}
The BLER of \((N,K, \mathcal{I})\) polar code, where \(\mathcal{I}=\{i_1,i_2,\cdots,i_K\}\) with \(i_1<i_2<\cdots<i_K\), under an SCL decoding, with list size \(L=2^m\) with \(m\in [\![1,+\infty]\!]\), over a symmetric B-DMC  \(\mathcal{W}: \mathcal{X} \to \mathcal{Y}\) satisfies 

\begin{align}
&P(\mathcal{E}_{\mathrm{SCL}}) 
= P(\mathcal{E}_{\mathrm{PL}})+P(\mathcal{E}_{\mathrm{PS}}) \notag\\
&= \sum_{\substack{i\in\mathcal{I}\\ i\ge i_m}}
\sum_{y_1^N\in\mathcal{Y}^N} \Pr(Y_1^N = y_1^N)\,\mathcal{P}^{\mathrm{SCL}}_{e,i}(y_1^N)
+ P(\mathcal{E}_{\mathrm{PS}}) \notag\\
&\le \sum_{\substack{i\in\mathcal{I}\\ i\ge i_m}}
\sum_{y_1^N\in\mathcal{Y}^N} \Pr(Y_1^N = y_1^N)\,\mathcal{P}^{\prime\mathrm{SCL}}_{e,i}(y_1^N)
+ P(\mathcal{E}_{\mathrm{PS}}),
\end{align}
where $i_m\in\mathcal{I}$ denote the $m$-th information index.
\end{proposition}

\begin{proof}
The result follows directly from the preceding analysis.
\end{proof}

Having derived a new formulation for the PL error probability in terms of the bit-wise path-elimination probabilities, we are now in a position to leverage this characterization to construct polar codes optimized for SCL decoding. The resulting SCL-oriented code construction, which aims to minimize the PL error probability, is presented in the next section.

\section{Polar Code Construction Optimized for SCL Decoding}
\label{sec:scl_cons}

In this section, we address the construction of an $(N,K)$ polar code optimized for SCL decoding over arbitrary symmetric B-DMCs \(\mathcal{W}:\mathcal{X} \to \mathcal{Y}\). The goal of polar code construction is to identify an information set $\mathcal{I}\subseteq[\![1,N]\!]$ that minimizes the BLER of the resulting code, i.e.,
\[
\mathcal{I}
=
\arg\min_{\substack{\mathcal{I}\subset[\![1,N]\!]\\|\mathcal{I}|=K}}
P(\mathcal{E}_{\mathrm{SCL}}).
\]
Under SCL decoding, this problem can be equivalently expressed as
\begin{align}
\mathcal{I}
&=
\arg\min_{\substack{\mathcal{I}\subset[\![1,N]\!]\\|\mathcal{I}|=K}}
\bigl[
P(\mathcal{E}_{\mathrm{PL}})
+
P(\mathcal{E}_{\mathrm{PS}})
\bigr].
\end{align}

This optimization problem is inherently combinatorial, and an exhaustive search over all candidate information sets becomes infeasible even for moderate block lengths. Since the PS error probability can be effectively reduced by employing a properly designed outer code (e.g., CRC~\cite{niu2012crc} or PAC~\cite{arikan2019sequential,moradi2021monte}), which preserves the minimum distance and, through an appropriate pre-transformation, reduces the number of minimum-distance codewords and increases the likelihood of selecting the correct codeword from the decoder’s list \cite{li2019pre}, this paper focuses on an SCL-optimized polar code construction that specifically aims to minimize the PL error probability.

\vspace{1mm}
\begin{algorithm}[b]
\caption{SCL-Optimized Polar Code Construction}
\label{alg:scl_cons}
\begin{algorithmic}[1]
\State \textbf{Input:} $N$, $K$, channel $\mathcal{W}$, list size $L$, target $E_s/N_0$
\State \textbf{Output:} Information set $\mathcal{I}$
\State Initialize $\mathcal{I}_0\gets[\![1,N]\!]$,
$K_0\gets N$, $t\gets 0$
\While{$K_t>K$}
    \State Estimate $\{\mathbb{E}[\mathcal{P}_{e,i}^{\prime\mathrm{SCL}}]\}_{i\in\mathcal{I}_t}$
    \State $i^* \gets \arg\max_{i\in\mathcal{I}_t}
    \mathbb{E}[\mathcal{P}_{e,i}^{\prime\mathrm{SCL}}]$
    \State $\mathcal{I}_{t+1}\gets\mathcal{I}_t\setminus\{i^*\}$
    \State $K_{t+1}\gets|\mathcal{I}_{t+1}|$, $t\gets t+1$
\EndWhile
\State \Return $\mathcal{I}\gets\mathcal{I}_t$
\end{algorithmic}
\end{algorithm}

The main objective of this work is to identify an information set that minimizes the PL error probability for a given SCL decoder with a specified list size. To this end, we adopt an incremental construction strategy analogous to Algorithm~\ref{alg:incremental_polar} to determine the information set $\mathcal{I}$, to reduce $P(\mathcal{E}_{\mathrm{PL}})$. A high-level description of the resulting SCL-oriented construction is provided in Algorithm~\ref{alg:scl_cons}. Algorithm~\ref{alg:scl_cons} incrementally freezes the least reliable bits,
thereby approximating the solution to the optimization problem
\begin{align*}
\mathcal{I}^*
=
\arg\min_{\substack{\mathcal{I}\subset[\![1,N]\!]\\|\mathcal{I}|=K}}
P(\mathcal{E}_{\mathrm{PL}}).
\end{align*}
Using the PL formulation derived in Section~\ref{sec:scl_bep}, this objective can be written as
\begin{align}
\label{eq:scl_pl_opt}
\mathcal{I}^*
=
\arg\min_{\substack{\mathcal{I}\subset[\![1,N]\!]\\|\mathcal{I}|=K}}
\sum_{i\in\mathcal{I}}
\sum_{y_1^N\in\mathcal{Y}^N}
\Pr(Y_1^N = y_1^N)\,\mathcal{P}_{e,i}^{\prime\mathrm{SCL}}(y_1^N),
\end{align}
where $\mathcal{P}_{e,i}^{\prime\mathrm{SCL}}(y_1^N)$ denotes the probability that the correct path is eliminated at decoding stage $i$, conditioned on the channel output $y_1^N$. The objective in~\eqref{eq:scl_pl_opt} constitutes an upper bound on $P(\mathcal{E}_{\mathrm{PL}})$, as it also includes indices $i \le i_m$, where $i_m$ denotes the $m$-th information bit. For these indices, no pruning occurs, and the denominator in~\eqref{eq:scl_llr} is undefined. In this case, the denominator is set to a small constant $\epsilon$, reflecting the negligible probability of eliminating the correct path at stages $i \le i_m$.  Furthermore, in computing $\mathcal{P}_{e,i}^{\prime\mathrm{SCL}}(y_1^N)$ in~\eqref{eq:pp_i_SCL}, the LLR values are clipped to the range $[-\mathbb{L}_{\max},\,\mathbb{L}_{\max}]$ to ensure numerical stability, with $\mathbb{L}_{\max}=30$ used throughout Algorithm~\ref{alg:scl_cons}.

For notational convenience, we define the PL error probability of each bit $i\in\mathcal{I}$ as
\begin{align}\label{eq:bit_err_scl}
\mathbb{E}\!\left[\mathcal{P}_{e,i}^{\prime\mathrm{SCL}}\right]
\triangleq
\sum_{y_1^N\in\mathcal{Y}^N}
\Pr(Y_1^N =y_1^N)\,\mathcal{P}_{e,i}^{\prime\mathrm{SCL}}(y_1^N).
\end{align}
The SCL-optimized polar code construction problem can then be expressed as
\begin{align*}
\mathcal{I}^*
=
\arg\min_{\substack{\mathcal{I}\subset[\![1,N]\!]\\|\mathcal{I}|=K}}
\sum_{i\in\mathcal{I}}
\mathbb{E}\!\left[\mathcal{P}_{e,i}^{\prime\mathrm{SCL}}\right].
\end{align*}
Algorithm~\ref{alg:scl_cons} formulates an SCL-optimized polar code construction and takes as input the code length $N$, target dimension $K$, list size $L$, channel \(\mathcal{W}\), and design $E_s/N_0$. It initializes a rate-one polar code with $\mathcal{I}_0\gets[\![1,N]\!]$ and $K_0=N$. At iteration $t$, SCL decoding with list size $L$ is employed to estimate \(\mathbb{E}\!\left[\mathcal{P}_{e,i}^{\prime\mathrm{SCL}}\right]\) for \(i \in \mathcal{I}_t\). Specifically, $M$ channel outputs $\{y_{1,d}^N\}_{d=1}^M$ are generated by transmitting a $(N,K_t,\mathcal{I}_t)$ polar code over the symmetric B-DMC $\mathcal{W}$. For each $i\in\mathcal{I}_t$, the quantity $\mathcal{P}_{e,i}^{\prime\mathrm{SCL}}(y_{1,d}^N)$ is computed using \eqref{eq:pp_i_SCL}, and the expectation
$\mathbb{E}\!\left[\mathcal{P}_{e,i}^{\prime\mathrm{SCL}}\right]$ is approximated as
\[
\mathbb{E}\!\left[\mathcal{P}_{e,i}^{\prime\mathrm{SCL}}\right]
\approx
\frac{1}{M}\sum_{d=1}^M
\mathcal{P}_{e,i}^{\prime\mathrm{SCL}}(y_{1,d}^N).
\]
These values serve as reliability metrics that quantify the likelihood that the correct path is eliminated at the decoding of the $i$-th information bit. To approximately solve~\eqref{eq:scl_pl_opt}, the information bit with the largest \(\mathbb{E}\!\left[\mathcal{P}_{e,i}^{\prime\mathrm{SCL}}\right]\) is identified and frozen, i.e.,
\[
i^*
=
\arg\max_{i\in\mathcal{I}_t}
\mathbb{E}\!\left[\mathcal{P}_{e,i}^{\prime\mathrm{SCL}}\right],
\qquad
\mathcal{I}_{t+1}\gets\mathcal{I}_t\setminus\{i^*\},
\]
with the code dimension updated as $K_{t+1}\gets|\mathcal{I}_{t+1}|$. This
procedure is repeated until the target dimension is reached, i.e., $K_t=K$.



\begin{figure*}[t]
    \centering
    \centering
    \includegraphics[width=0.9\textwidth]{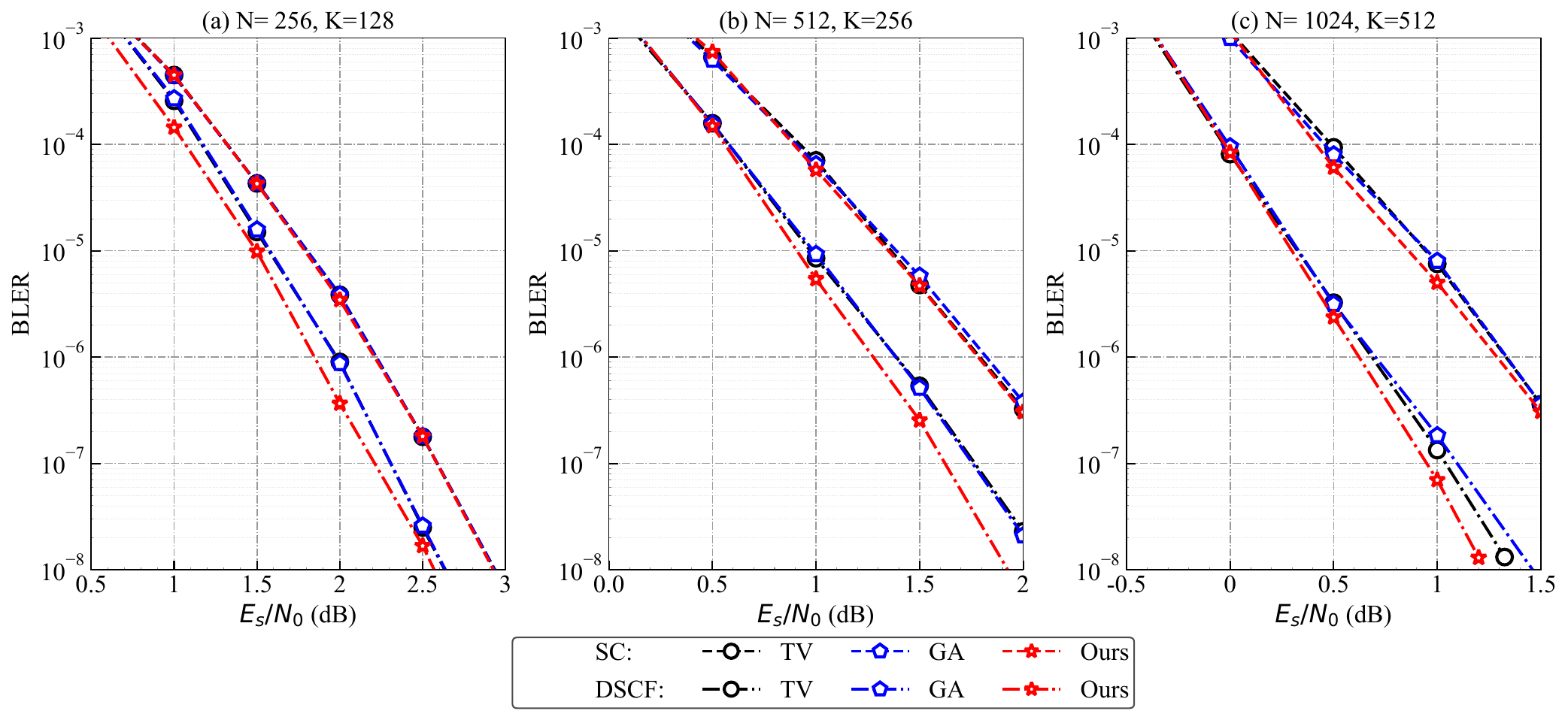}
    \caption{BLER comparison of the proposed polar codes optimized for SC decoding and benchmark designs over an AWGN channel. Subfigures~(a)–(c) correspond to $(N,K)=(256,128)$, $(512,256)$, and $(1024,512)$, respectively. The proposed codes are compared with GA-based and Tal-Vardy (TV) constructions under SC and DSCF decoding. For DSCF, a bit-flip order $\omega=1$ is used with up to $T=10$ attempts for $(256,128)$ and $T=8$ attempts for larger block lengths, employing 16-bit and 24-bit CRCs, respectively.}
    \label{fig:P256_128}
\end{figure*}

\section{Numerical Results}\label{sec:sim}

This section evaluates the performance of the proposed polar code construction algorithms. Subsection~\ref{subsec:perf} compares their error-correction performance with existing designs, including GA-based~\cite{trifonov2012efficient}
and Tal--Vardy~\cite{tal2013construct} constructions for SC decoding, and RM-polar~\cite{li2014rm}, GA-based~\cite{trifonov2012efficient}, and 5G polar codes~\cite{bioglio2020design} for SCL decoding. Subsection~\ref{subsec:cons}
illustrates the resulting information and frozen bit distributions, while Subsection~\ref{subsec:comp} compares the computational complexity of the proposed methods with existing approaches.

\subsection{Performance Analysis}\label{subsec:perf}


We first evaluate polar codes constructed using Algorithm~\ref{alg:incremental_polar}, optimized for SC decoding, and compare
their performance with established SC-optimized designs. We then assess the proposed SCL-optimized construction by benchmarking it against 5G, GA-based, and RM-polar polar codes. All simulations assume transmission of the all-zero codeword and are run until at least 100 block errors are observed. The design $E_s/N_0$ values for benchmark codes are chosen according to the literature to ensure competitive performance in the BLER range of $10^{-3}$ to $10^{-8}$. SCL decoding simulations are
performed using the AFF3CT C++ toolbox~\cite{cassagne2019aff3ct}.

\subsubsection{Performance of Polar Codes Optimized for SC Decoding}
To evaluate the effectiveness of the proposed polar code construction algorithm optimized for SC decoding, we apply it to design polar codes for communication over both AWGN and BSC channels. The resulting codes are compared with established SC-optimized constructions, including the Tal and Vardy (TV) polar codes~\cite{tal2013construct} and GA-based polar codes~\cite{trifonov2012efficient}. The TV polar codes are optimized at \(E_s/N_0 = 0\)~dB, as reported in~\cite{channelcodes}. The performance of the proposed and benchmark codes is evaluated under both SC and DSCF decoding~\cite{chandesris2018dynamic}. For DSCF decoding, the bit-flip order is set to \(\omega = 1\), with a maximum of \(T = 10\) additional decoding attempts for the \((256,128)\) polar code and \(T = 8\) for the \((512,256)\) and \((1024,512)\) codes. Additional decoding attempts in DSCF are triggered only when the initial SC decoding fails the CRC check, using the 5G 16-bit CRC for the \((256,128)\) code and the 5G 24-bit CRC for the \((512,256)\) and \((1024,512)\) codes.

We first consider a $(256,128)$ polar code designed at $E_s/N_0=1.2$~dB and compare it with a TV code designed at $E_s/N_0=0$~dB~\cite{channelcodes} and a GA-based code designed at $E_s/N_0=3$~dB. As shown in Fig.~\ref{fig:P256_128}(a), the proposed construction achieves performance comparable to the benchmark designs under SC decoding. Under DSCF decoding, however, it consistently outperforms both baselines across the entire BLER range, yielding up to a $0.15$~dB gain at $\text{BLER}=10^{-6}$.

\begin{figure*}[t]
    \centering
    \centering
    \includegraphics[width=0.55\textwidth]{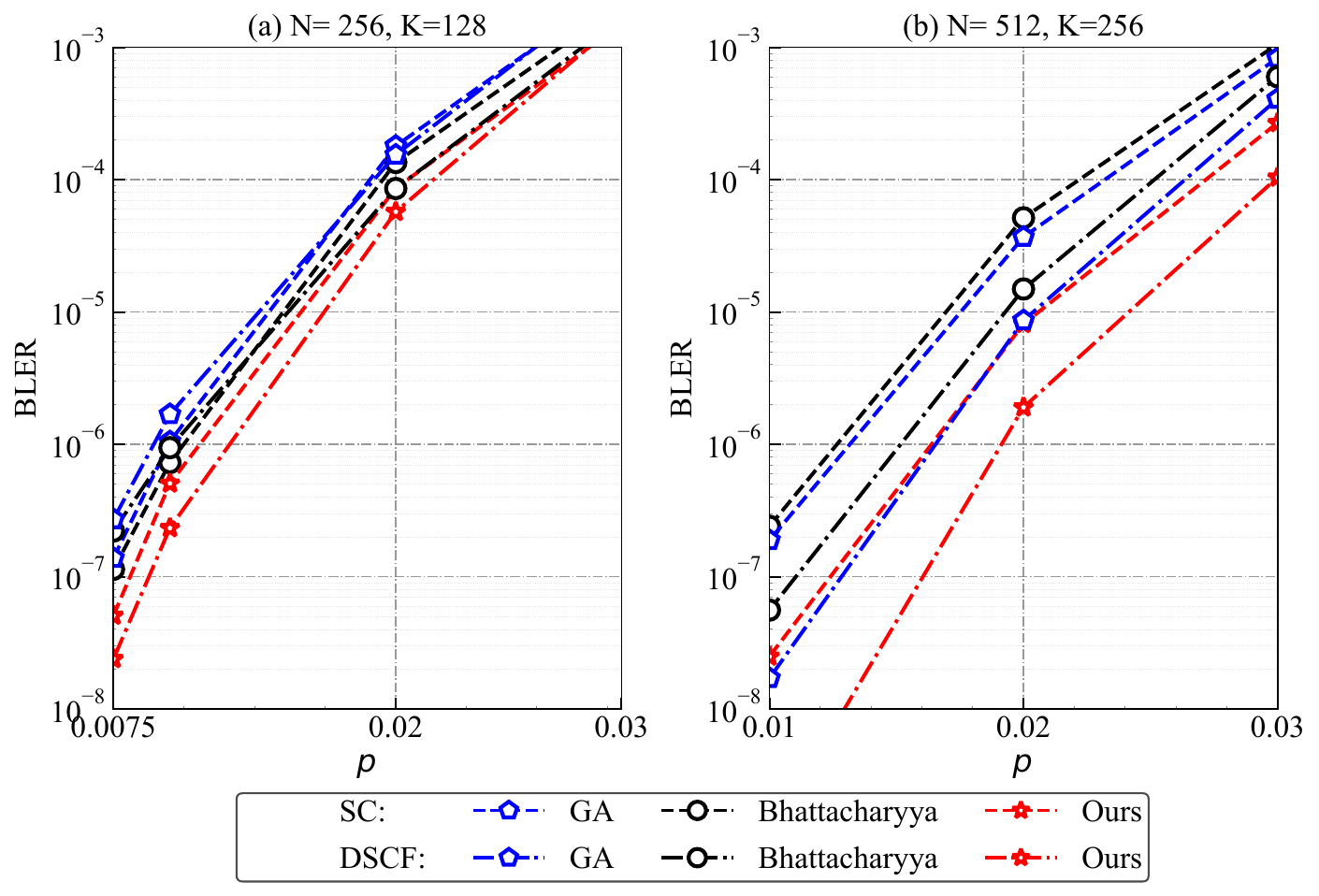}
    \caption{BLER comparison of the proposed polar codes optimized for SC decoding and benchmark designs over a BSC channel. Subfigure~(a) shows results for $(N,K)=(256,128)$, where the proposed code is compared with GA-based and Bhattacharyya-based polar codes~\cite{arikan2009channel} under SC and DSCF decoding ($\omega=1$, $T=10$) using a 16-bit CRC. Subfigure~(b) reports results for $(N,K)=(512,256)$, where the proposed code is compared with the same benchmarks under SC and DSCF decoding ($\omega=1$, $T=10$) using a 24-bit CRC.}
    \label{fig:fer_bsc}
\end{figure*}

\begin{figure*}[t]
    \centering
    \includegraphics[width=0.9\textwidth]{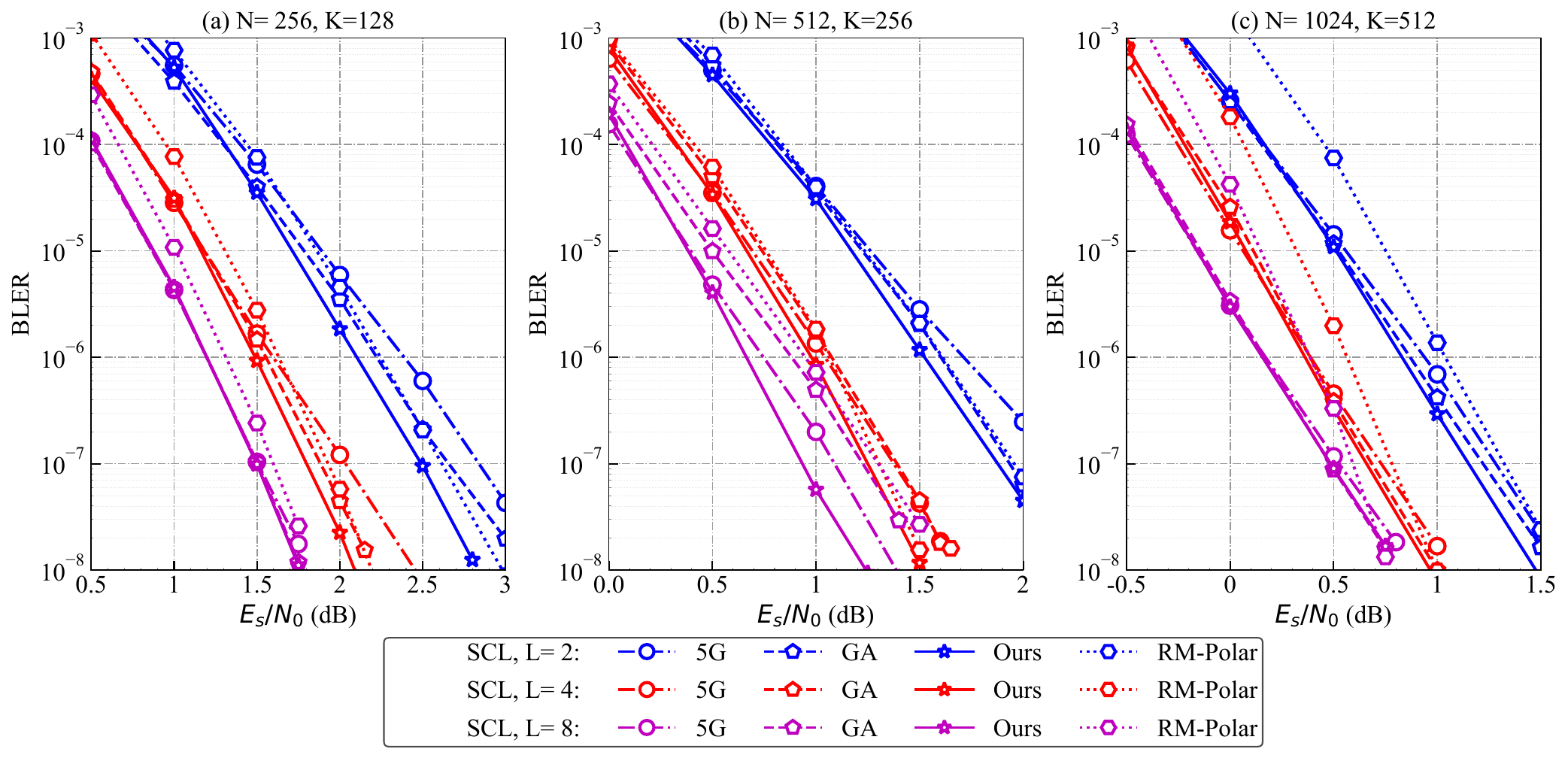}

\caption{BLER comparison of the proposed polar codes optimized for SCL decoding and benchmark designs over an AWGN channel. Subfigure~(a) shows results for $(N,K)=(256,128)$ with a 16-bit CRC, where the proposed codes are compared with
GA-based, RM-polar, and 5G polar codes under SCL decoding with $L=2,4,8$. Subfigure~(b) reports results for $(N,K)=(512,256)$ with a 24-bit CRC, where the proposed codes are compared with GA-based, RM-polar, and 5G polar codes. Subfigure~(c) shows results for $(N,K)=(1024,512)$ with a 24-bit CRC, where the proposed codes are compared with GA-based, RM-polar, and 5G polar codes under SCL decoding with $L=2,4,8$.}
    \label{fig:scl_fer}
\end{figure*}

Next, we evaluate a $(512,256)$ polar code designed at $E_s/N_0=0.5$~dB.
Fig.~\ref{fig:P256_128}(b) shows that, under SC decoding, the proposed code
maintains performance comparable to the TV and GA-based designs, in agreement
with the $(256,128)$ results. Under DSCF decoding, the proposed construction
provides a clear advantage, achieving up to a $0.2$~dB gain at
$\text{BLER}=2\times10^{-8}$.

Finally, Fig.~\ref{fig:P256_128}(c) reports the results for a $(1024,512)$ polar
code designed at $E_s/N_0=0$~dB. Under SC decoding, the proposed construction
matches the performance of the TV and GA-based codes over the entire simulated
range, confirming its scalability to larger block lengths. Under DSCF decoding,
it again delivers superior performance, achieving up to a $0.2$~dB gain at
$\text{BLER}=10^{-7}$ using the same DSCF configuration ($\omega=1$, $T=8$).

To further assess the proposed construction across different channel models, we
apply the algorithm to design polar codes for transmission over a BSC. As in the
AWGN case, decoding is performed using both SC and DSCF, with $\omega=1$ and
$T=10$ for all configurations. 
We first consider a $(256,128)$ polar code designed for a BSC with crossover
probability $p=0.02$, and compare it with a GA-based code designed at
$E_s/N_0=1.2$~dB and the Bhattacharyya-based construction of~\cite{arikan2009channel}
designed for $p=0.008$. As shown in Fig.~\ref{fig:fer_bsc}(a), the proposed
construction outperforms both benchmarks under SC and DSCF decoding, with
particularly pronounced gains under DSCF.

We then evaluate a $(512,256)$ polar code designed for $p=0.025$, compared
against GA-based and Bhattacharyya-based codes designed at $E_s/N_0=0$~dB and
$p=0.015$, respectively. Fig.~\ref{fig:fer_bsc}(b) shows that the proposed
construction again achieves the best BLER performance, especially under DSCF
decoding. These results highlight the flexibility of the proposed algorithm and
its effectiveness in constructing polar codes across different channel models.

\subsubsection{Performance of Polar Codes Optimized for SCL Decoding} 
To evaluate the effectiveness of the proposed SCL-optimized construction
(Algorithm~\ref{alg:scl_cons}), we benchmark the resulting polar codes against
state-of-the-art SCL-oriented designs, including 5G polar codes~\cite{bioglio2020design},
RM--polar codes~\cite{li2014rm}, and GA-based polar codes. For fair comparison,
the benchmark codes are designed at $E_s/N_0$ values selected to ensure strong
performance over the simulated BLER range. SCL decoding is evaluated for list
sizes $L=2,4,8$, using a 16-bit 5G CRC for $(256,128)$ and a 24-bit 5G CRC for
$(512,256)$ and $(1024,512)$.

We first consider a $(256,128)$ polar code. The proposed codes are designed at
$E_s/N_0=1.2$~dB for $L=2$ and $L=8$, and at $E_s/N_0=0.8$~dB for $L=4$, and are
compared with RM--polar, GA-based, and 5G polar codes. As shown in
Fig.~\ref{fig:scl_fer}(a), the proposed construction consistently outperforms the
benchmarks. At $\text{BLER}=10^{-7}$, gains of about $0.4$~dB over the 5G code and
$0.15$–$0.2$~dB over RM--polar and GA-based codes are observed for $L=2$. For
$L=4$, the proposed code provides roughly a $0.35$~dB gain over the 5G code and
$0.15$~dB over RM--polar and GA-based codes at $\text{BLER}=2\times10^{-8}$. For
$L=8$, it achieves the best or comparable performance across the simulated
range, with about a $0.1$~dB gain over RM--polar at $\text{BLER}=3\times10^{-7}$.

We next evaluate a $(512,256)$ polar code, with all proposed designs constructed
at $E_s/N_0=0.4$~dB and compared against RM-polar and GA-based polar codes, each optimized at
$E_s/N_0=0$~dB, as well as 5G standard polar codes. Fig.~\ref{fig:scl_fer}(b) shows that for $L=2$, the proposed code
achieves gains of up to $0.38$~dB over the 5G code and about $0.1$~dB over RM-polar
and GA-based codes at $\text{BLER}=10^{-7}$. For $L=4$, gains of approximately
$0.2$~dB over the 5G and GA-based codes and up to $0.15$~dB over RM-polar are
observed, with the proposed $L=4$ code outperforming RM-polar and GA-based
codes decoded with $L=8$ in the low-BLER regime. For $L=8$, the proposed code
maintains clear advantages, achieving about $0.2$~dB gain over the 5G code and
$0.35$–$0.4$~dB gains over RM-polar and GA-based codes at
$\text{BLER}=10^{-7}$.

Finally, we consider a $(1024,512)$ polar code, with all designs constructed at
$E_s/N_0=0$~dB. As shown in Fig.~\ref{fig:scl_fer}(c), the proposed construction
achieves comparable or superior performance for all list sizes. At
$\text{BLER}=10^{-7}$, gains of about $0.2$~dB over the 5G code and $0.25$ and
$0.1$~dB over RM-polar and GA-based codes are observed for $L=2$. For $L=4$, the
proposed code provides up to $0.1$~dB gain over the 5G code and $0.2$~dB over
RM-polar, while for $L=8$ it consistently matches or outperforms the benchmark
designs, achieving $0.05$–$0.1$~dB gains over the 5G code and about $0.2$~dB over
RM-polar.

Overall, these results demonstrate that the proposed construction reliably produces polar codes well matched to SCL decoding with different list sizes, yielding consistent BLER improvements without increasing decoding complexity, solely through optimized selection of information and frozen bits.

\begin{figure}[t]
    \centering
    \includegraphics[width=0.485\textwidth]{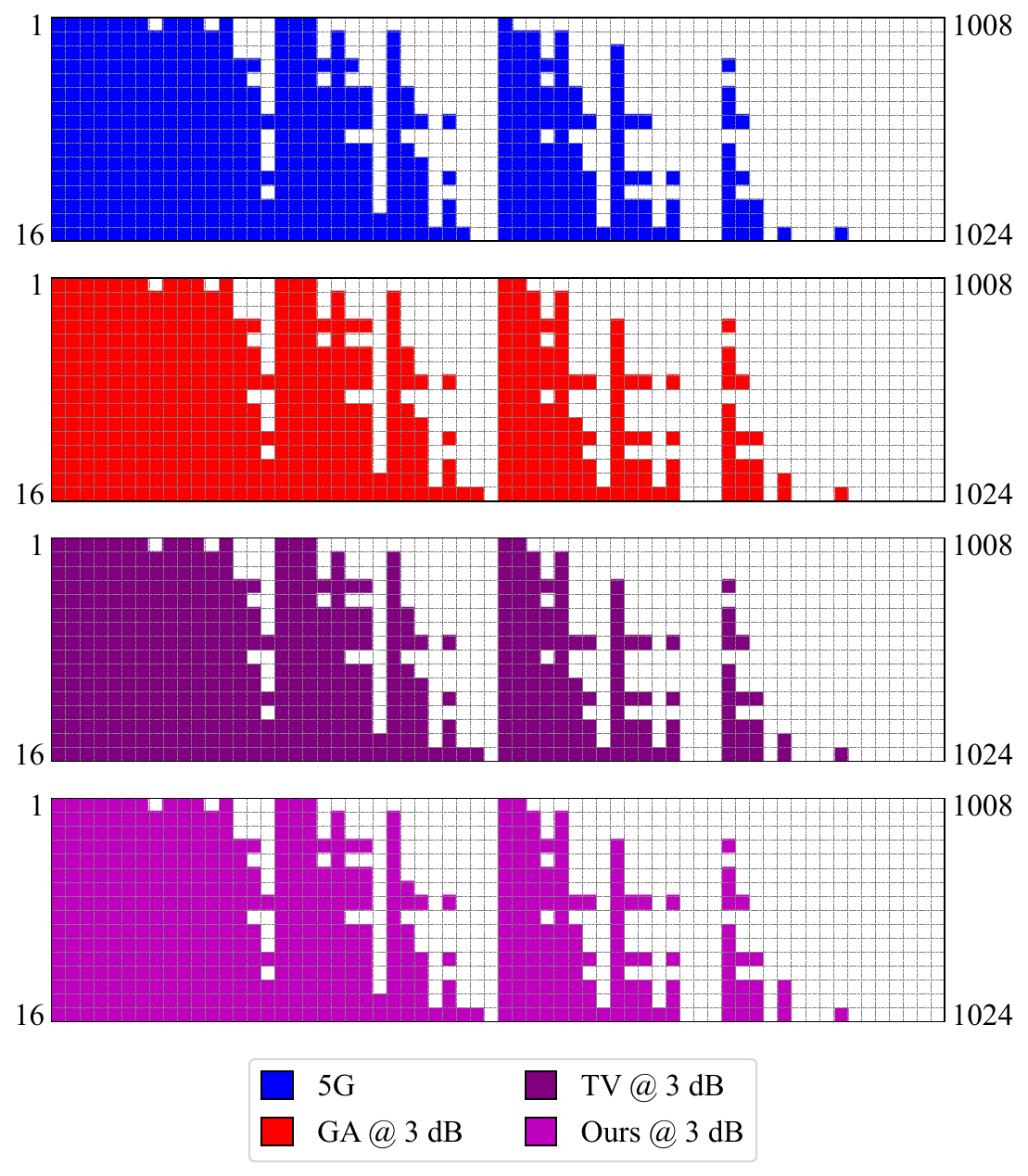}
    \caption{Bit distributions of the 5G, GA-based, TV, and proposed polar codes with parameters \(N = 1024\) and \(K = 512\). Colored pixels represent frozen bits, while white pixels represent information bits.}
    \label{fig:con_1024_sc}
\end{figure}

\subsection{Construction Analysis}\label{subsec:cons}

In this section, we compare the distributions of frozen and information bits in
the proposed constructions with those of existing benchmarks. The discussion is
organized into two parts. First, the SC-optimized construction is compared with
TV, GA-based, and 5G polar codes (the latter included for reference only, as it is
not optimized for SC decoding). Second, the SCL-optimized constructions with
list sizes $L=2,4,8$ are compared with RM-polar, GA-based, and 5G polar codes.

To visualize the bit distributions, frozen and information bits are represented
by a binary vector, where $0$ and $1$ denote frozen and information bits,
respectively. This vector is reshaped into a $d\times b$ matrix, with each column
corresponding to $d$ consecutive bit channels, and displayed as an image using
colored pixels for frozen bits and white pixels for information bits.

\begin{figure}
    \centering
    \includegraphics[width=0.485\textwidth]{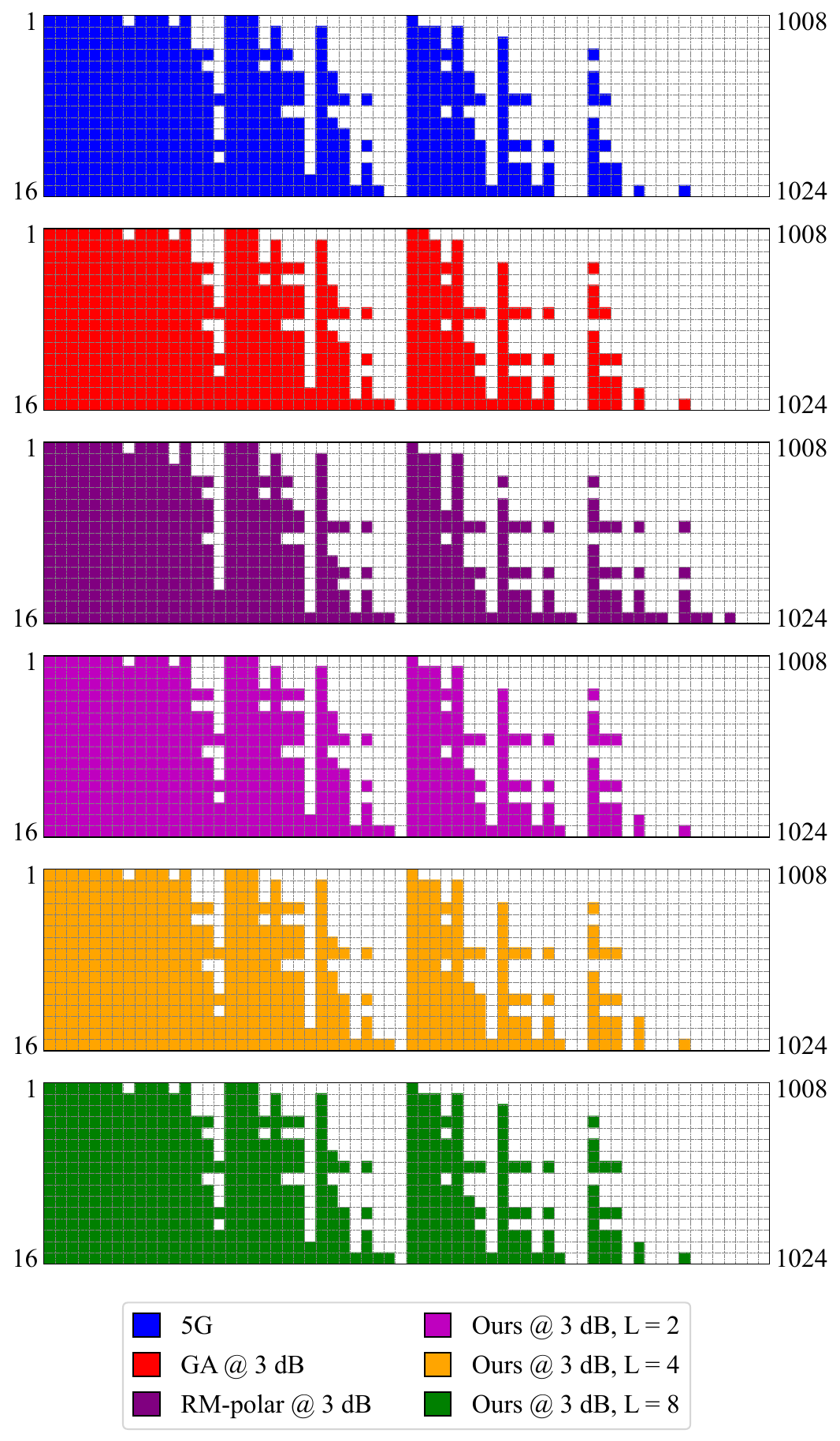}
    \caption{Bit distributions of the 5G standard, GA-based, RM–polar, and proposed polar codes optimized for list sizes \(L = 2, 4, 8\) with parameters \(N = 1024\) and \(K = 512\). Colored pixels represent frozen bits, while white pixels represent information bits.}
    \label{fig:bitmask_1024_SCL}
\end{figure}

\subsubsection{Polar code construction optimized for SC decoding}
Figure~\ref{fig:con_1024_sc} shows the bit-channel distributions of the proposed,
TV, GA-based, and 5G polar codes for $(N,K)=(1024,512)$. The proposed and benchmark
codes correspond to those evaluated in Fig.~\ref{fig:P256_128}(c), with the 5G
polar code included for reference to highlight structural differences.
Noticeable discrepancies are observed between the proposed construction and
existing designs. In particular, the proposed method freezes several high-index
bit channels while assigning some low-index channels as information bits, in
contrast to the monotonic reliability ordering typically exhibited by the
benchmark constructions.

\subsubsection{Polar code construction optimized for SCL decoding}
Figure~\ref{fig:bitmask_1024_SCL} illustrates the bit-channel distributions of the
proposed SCL-optimized polar codes for list sizes $L=2,4,8$, compared with the
5G, GA-based, and RM-polar codes. All codes use $(N,K)=(1024,512)$ and correspond
to those evaluated in the BLER results presented earlier.


The proposed SCL-optimized constructions exhibit adaptive bit allocations that
depend on the list size. Across all $L$, the algorithm freezes several mid- and
high-index channels that remain information bits in GA-based and 5G designs,
while still assigning highly reliable high-index channels as information bits.
This behavior distinguishes the proposed construction from both DE-based and
RM-polar approaches.

\subsection{Complexity Analysis}\label{subsec:comp}

To analyze the complexity of the proposed polar code constructions (Algorithms~\ref{alg:incremental_polar} and~\ref{alg:scl_cons}), we focus on the complexity analysis of Algorithm~\ref{alg:scl_cons}, which is optimized for SCL decoding with a list size \(L\). The complexity of Algorithm~\ref{alg:incremental_polar} can then be readily derived as a special case of this analysis by setting \(L = 1\).

The proposed polar code construction algorithms (Algorithm~\ref{alg:scl_cons}) incrementally identify and freeze the least reliable bit-channels based on the estimated PS error probability \( \mathbb{E}\left[\mathcal{P}^{\prime SCL}_{e,i}\right] \) defined in~(\ref{eq:bit_err_scl}). The overall computational complexity of the algorithm arises from three primary components: the estimation of \( \mathbb{E}\left[\mathcal{P}^{\prime SCL}_{e,i}\right] \) for each bit-channel, the iterative freezing and updating of the information and frozen sets, and the decoding operations required during each reliability evaluation. For each bit-channel \( i \in \mathcal{I} \), the algorithm estimates \( \mathbb{E}\left[\mathcal{P}^{\prime SCL}_{e,i}\right] \) by averaging over \( M \) channel realizations, expressed as 
\[
\mathbb{E}\left[\mathcal{P}^{\prime SCL}_{e,i}\right] \approx \frac{1}{M} \sum_{d=1}^{M} \mathcal{P}_{e,i}^{\prime\text{SCL}}\!\left(y_1^{N,d}\right),
\]
where \( \mathcal{P}_{e,i}^{\prime\text{SCL}}\!\left(y_1^{N,d}\right) \) is computed based on the decoder’s LLR outputs. Each evaluation of \( \mathcal{P}_{e,i}^{\prime\text{SCL}}\!\left(y_1^{N,d}\right) \) requires one SCL with list size \(L\) decoding pass, with a computational complexity of \( \mathcal{O}(LN \log N) \). Therefore, the total cost of evaluating all \( N \) bit-channels once over \( M \) channel realizations is \( \mathcal{O}(M L N \log N) \).

The algorithm performs \( N - K \) freezing iterations, where at each iteration \( t \), one bit index \( i^* \) with the highest estimated \( \mathbb{E}\left[\mathcal{P}^{\prime SCL}_{e,i^*}\right]\) is moved from the information set \(\mathcal{I}_t\) to the frozen set \(\mathcal{I}^c_t\). Since bit-channel reliabilities are updated or re-estimated at each step, the total number of reliability evaluations scales linearly with the number of frozen bits, resulting in an overall computational complexity of \( \mathcal{O}\left((N - K) M L N \log N\right) \). 

For comparison, the GA-based construction~\cite{trifonov2012efficient} has a computational complexity of \( \mathcal{O}(N \log N) \), while the TV algorithm~\cite{tal2013construct} exhibits a complexity of \( \mathcal{O}(\mu N \log N) \), where \(\mu\) denotes the output alphabet size used in the degrading/merging procedure of density evolution. Although Algorithm~\ref{alg:incremental_polar} has a higher complexity than GA-based and TV methods, it offers several key advantages. First, it utilizes decoder-based reliability metrics derived from LLRs, thereby capturing practical decoding error correlations neglected by DE-based approaches. Second, it is channel-agnostic and applicable to arbitrary symmetric B-DMCs without requiring analytical channel models. Furthermore, the polar codes constructed using the proposed algorithm outperform GA-based and TV codes under DSCF decoding.

The most comparable polar code construction to the proposed method is the MC–based approach, in which the BLER of each candidate code is estimated by decoding \( M \) simulated samples. If the target BLER at a given \( E_s/N_0 \) is \( \varepsilon \), the MC method typically requires \( M \gg 1 / \varepsilon \) samples to obtain an accurate estimate of the BLER. Since there are \( \binom{N}{K} \) possible code constructions to evaluate, the overall computational complexity is given by~\cite{liao2021construction}
\[
\mathcal{O}\!\left( L N \log N \cdot M \cdot \binom{N}{K} \right).
\]
Because \( N - K \ll \binom{N}{K} \), the computational complexity of the proposed polar code construction algorithm is several orders of magnitude lower than that of MC-based approaches.


\section{Conclusion}\label{sec:con}

This paper derived a refined semi-analytical expression for estimating the BLER of polar codes under SC and SCL decoding. For SC decoding, a new expression derived from the decoder’s LLRs enabled a recursive construction method that selects frozen bits based on their estimated reliability. For SCL decoding, the BLER was decomposed into PL and PS error components, and a semi-analytical PL expression was developed to construct polar codes optimized for a given list size. Simulation results showed that the proposed SC-optimized designs achieved up to a 0.2~dB performance gain under DSCF decoding compared to benchmark constructions over AWGN channels and exhibited superior performance on BSC channels. For SCL-optimized polar codes, the proposed method achieved comparable or better performance across all list sizes, with gains of up to 0.4~dB relative to benchmark designs.

These findings affirm the potential of our method to advance practical polar code design. Future research could focus on further reducing the computational complexity of the proposed construction method.

\bibliographystyle{IEEEtran}
\bibliography{ref}

\end{document}